\documentclass[letterpaper,journal]{IEEEtran}
\usepackage{amsmath,amsfonts}
\usepackage{algorithmic}
\usepackage{algorithm}
\usepackage{array}
\usepackage[caption=false,font=footnotesize]{subfig}
\usepackage{textcomp}
\usepackage{stfloats}
\usepackage{url}
\usepackage{verbatim}
\usepackage{graphicx}
\usepackage{cite}
\usepackage{enumitem}
\usepackage{amsmath,amsfonts,amsthm,amssymb,mathtools}

\usepackage{microtype}
\usepackage{xfrac}
\usepackage[makeroom]{cancel}
\usepackage{mathtools}
\usepackage{bookmark}
\usepackage{enumitem}
\usepackage{theoremref}
\usepackage{hyperref}
\usepackage[most,many,breakable]{tcolorbox}
\usepackage{varwidth}
\usepackage{varwidth}
\usepackage{etoolbox}
\usepackage{nameref}
\usepackage{multicol,array}
\usepackage{algorithm,algorithmic}
\usepackage{comment} 
\usepackage{import}
\usepackage{xifthen}
\usepackage{pdfpages}
\usepackage{transparent}
\usepackage{chngcntr}
\usepackage{tikz}
\usetikzlibrary{calc,decorations.pathreplacing}
\usepackage{titletoc}

\usepackage{cleveref}
\crefname{figure}{Fig.}{Figs.}
\Crefname{figure}{Fig.}{Figs.}
\crefname{subfigure}{Fig.}{Figs.}
\Crefname{subfigure}{Fig.}{Figs.}
\crefname{table}{TABLE}{TABLES}
\Crefname{table}{TABLE}{TABLES}
\crefname{section}{Section}{Sections}
\Crefname{section}{Section}{Sections}
\crefname{subsection}{Section}{Sections}
\Crefname{subsection}{Section}{Sections}
\crefname{algorithm}{Algorithm}{Algorithms}
\Crefname{algorithm}{Algorithm}{Algorithms}
\crefname{theorem}{Theorem}{Theorems}
\Crefname{theorem}{Theorem}{Theorems}
\crefname{proposition}{Proposition}{Propositions}
\Crefname{proposition}{Proposition}{Propositions}
\crefname{corollary}{Corollary}{Corollaries}
\Crefname{corollary}{Corollary}{Corollaries}
\crefname{appendix}{Appendix}{Appendices}
\Crefname{appendix}{Appendix}{Appendices}

\usepackage{empheq}

\definecolor{doc}{RGB}{0,60,110}
\definecolor{myg}{RGB}{56, 140, 70}
\definecolor{myb}{RGB}{45, 111, 177}
\definecolor{myr}{RGB}{199, 68, 64}
\definecolor{mytheorembg}{HTML}{F2F2F9}
\definecolor{mytheoremfr}{HTML}{00007B}
\definecolor{mylemmabg}{HTML}{FFFAF8}
\definecolor{mylemmafr}{HTML}{983b0f}
\definecolor{mypropbg}{HTML}{f2fbfc}
\definecolor{mypropfr}{HTML}{191971}
\definecolor{myexamplebg}{HTML}{F2FBF8}
\definecolor{myexamplefr}{HTML}{88D6D1}
\definecolor{myexampleti}{HTML}{2A7F7F}
\definecolor{mydefinitbg}{HTML}{E5E5FF}
\definecolor{mydefinitfr}{HTML}{3F3FA3}
\definecolor{notesgreen}{RGB}{0,162,0}
\definecolor{myp}{RGB}{197, 92, 212}
\definecolor{mygr}{HTML}{2C3338}
\definecolor{myred}{RGB}{127,0,0}
\definecolor{DodgerBlue}{RGB}{30, 144, 255}
\definecolor{myyellow}{RGB}{169,121,69}
\definecolor{myexercisebg}{HTML}{F2FBF8}
\definecolor{myexercisefg}{HTML}{88D6D1}

\hypersetup{
	colorlinks=true,
	linkcolor=myr,  
	citecolor=myg,  
	urlcolor=myb,   
	breaklinks=true,
	bookmarksnumbered=true,
	bookmarksopen=true,
	pdfdisplaydoctitle=true,
	pdfstartview={FitH},
	pdfpagemode=UseOutlines,
	unicode=true
}

\newif\ifMasterColorSwitch
\MasterColorSwitchtrue

\newif\ifEnableMyColors
\EnableMyColorstrue

\NewDocumentCommand{\MasterColorMode}{m}{%
	\ifthenelse{\equal{#1}{on}}%
	{%
		\MasterColorSwitchtrue
		\typeout{>>> Master Color Mode Enabled}%
	}%
	{%
		\MasterColorSwitchfalse
		\typeout{>>> Master Color Mode Disabled}%
	}%
}

\NewDocumentCommand{\ColorMode}{m}{%
	\ifthenelse{\equal{#1}{on}}%
	{%
		\EnableMyColorstrue
		\typeout{>>> Custom Colors Enabled}%
	}%
	{%
		\EnableMyColorsfalse
		\typeout{>>> Custom Colors Disabled}%
	}%
}

\newcommand{\RR}[1][]{\ensuremath{\ifstrempty{#1}{\mathbb{R}}{\mathbb{R}^{#1}}}}

\newcommand{\CC}[1][]{\ensuremath{\ifstrempty{#1}{\mathbb{C}}{\mathbb{C}^{#1}}}}
\newcommand{\PP}[1][]{\ensuremath{\ifstrempty{#1}{\mathbb{P}}{\mathbb{P}^{#1}}}}

\newcommand{\EE}{\ensuremath{\mathbb{E}}}

	\newcommand{\bfD}{\mathbf{D}}

\newcommand{\bfI}{\mathbf{I}}	
	
\newcommand{\bfM}{\mathbf{M}}

\newcommand{\bfa}{\mathbf{a}}	
	
\newcommand{\bfe}{\mathbf{e}}

\newcommand{\bfu}{\mathbf{u}}	\newcommand{\bfv}{\mathbf{v}}
	\newcommand{\bfx}{\mathbf{x}}
	\newcommand{\bfz}{\mathbf{z}}

\newcommand{\rmd}{\mathrm{d}} 
\newcommand{\rme}{\mathrm{e}} 

\newcommand{\veps}{\varepsilon}

\newcommand{\iid}{\textit{i.i.d.\ }}
\newcommand{\bfSigma}{\boldsymbol{\Sigma}}

\graphicspath{{./figs/}}

\newcommand{\ALGSTAGE}[1]{\item[]\textbf{#1}}

\newtheorem{theorem}{Theorem}

\newtheorem{proposition}{Proposition}
\newtheorem{corollary}{Corollary}

\begin{document}

\title{Deflection-Optimal Spectral Design for Diagonal Screening in Sparse Phase Retrieval Initialization}

\author{Mengchu~Xu,~\IEEEmembership{Member,~IEEE},
	and~Yonina~C.~Eldar,~\IEEEmembership{Fellow,~IEEE}
	\thanks{Mengchu Xu and Yonina C. Eldar are with the Faculty of Mathematics and Computer Science, Weizmann Institute of Science, Rehovot 7610001, Israel (e-mail: mengchu.xu@weizmann.ac.il; yonina.eldar@weizmann.ac.il). Yonina C. Eldar is also with the Department of Electrical and Computer Engineering, Northeastern University, Boston, MA 02115 USA (e-mail: y.eldar@northeastern.edu). Corresponding author: Mengchu Xu.}
}

\maketitle


\begin{abstract}
	Spectral initialization is a critical yet challenging step in sparse phase retrieval. Existing spectral design theory is largely tailored to dense phase retrieval, where the objective is eigenvector estimation. In contrast, sparse initialization first requires a statistically distinct support screening step whose design remains much less understood. This paper develops a stage-specific design theory for diagonal support screening. We formulate each coordinate score as a scalar statistic for distinguishing support from non-support coordinates and adopt the deflection criterion as a tractable measure of screening quality. Within a Hilbert-space formulation, we characterize the optimal spectral preprocessors that maximize this criterion. In the Gaussian model, the unique optimum is the centered linear preprocessor. To obtain a bounded implementation, we introduce a spherical normalization and characterize its exact optimal preprocessor. Since the exact spherical optimum exhibits a boundary singularity, we construct a bounded surrogate preprocessor and establish its unique optimality under a surrogate deflection criterion. The surrogate optimum is shown to be the direction-only projection of the Gaussian rule, removing the unbounded radial factor while preserving the same first-order screening structure. We further establish a general finite-sample diagonal bridge that connects the exact and surrogate deflection quotients to the initialization sample complexity, and that replacing the unknown signal energy by its empirical estimate introduces only a lower-order perturbation. Numerical experiments are consistent with the ordering predicted by the design quotients and show that the Gaussian centered rule and its spherical counterpart behave nearly identically at both the screening and initialization levels.
\end{abstract}


\section{Introduction}\label{sec:introduction}
Phase retrieval seeks to recover an unknown signal from magnitude-only measurements, a problem arising in coherent diffraction imaging, optics, and X-ray crystallography~\cite{Fienup1978Reconstructionobjecta, Miao1999ExtendingmethodologyXray, Harrison1993Phaseproblema, Shechtman2015Phaseretrievalapplication}. In the sparse regime, recovery requires estimating both the signal coefficients and their support, introducing an additional combinatorial challenge~\cite{Shechtman2014GESPAREfficientphase, Moravec2007Compressivephase, Ohlsson2012CPRLextensioncompressive, Li2013Sparsesignalrecovery}. We consider the complex Gaussian measurement model~\cite{Balan2006signalreconstructiona, Candes2013PhaseLiftExactstable, Candes2014Solvingquadraticequations}:
\begin{equation}\label{eq:model}
	y_i = |\bfa_i^* \bfx|^2,\qquad i=1,\dots,m,
\end{equation}
where \(\bfx\in\CC^n\) is \(k\)-sparse, the sensing vectors \(\bfa_i\) are independent standard complex Gaussian vectors, and the recovery target is \(\{\rme^{\mathbf{i}\theta}\bfx:\theta\in[0,2\pi)\}\). Modern sparse phase retrieval methods are typically nonconvex or multistage~\cite{Li2013Sparsesignalrecovery, Schniter2015Compressivephaseretrieval}, and both theory and empirical evidence suggest that their performance depends critically on the quality of the initial signal estimate~\cite{Li2022Samplingcomplexity,Sun2026Phaseretrieval,Huang2026Recoveryperformance, Chen2017Solvingrandomquadratic}. Understanding the fundamental measurement requirement for obtaining a high-quality initializer therefore becomes an important problem. This motivates the study of the sample complexity of initialization: how many samples are required to construct a sufficiently accurate initializer, and what are the corresponding asymptotic scaling law and leading constant?

\begin{figure}[!t]
	\centering
	\resizebox{\columnwidth}{!}{%
		\begin{tikzpicture}[
        font=\footnotesize,
        >=stealth,
        mainarrow/.style={->, line width=0.6pt, draw=black!75},
        feedback/.style={->, dashed, line width=0.6pt, draw=black!50},
        stage/.style={
                draw=black,
                line width=0.75pt,
                fill=white,
                rounded corners=2pt,
                inner xsep=5pt,
                inner ysep=5pt,
                align=center,
                text width=3.8cm
            },
        support/.style={
                draw=black!25,
                line width=0.5pt,
                fill=black!2,
                rounded corners=1pt,
                inner xsep=4pt,
                inner ysep=2pt,
                align=center,
                font=\scriptsize
            },
        brace/.style={
                decorate,
                decoration={brace, amplitude=3pt},
                draw=black!75,
                line width=0.5pt
            },
        sidenote/.style={
                align=left,
                text width=3.05cm,
                font=\scriptsize,
                text=black
            },
        loopinfo/.style={
                fill=white,
                inner xsep=2pt,
                inner ysep=1pt,
                align=center,
                font=\scriptsize,
                text=black!75,
                text width=2.20cm
            }
    ]
    \node[align=center] (input) at (0,0)
    {Observations $\{y_i\}_{i=1}^m$ and sensing vectors $\{\bfa_i\}_{i=1}^m$};

    \node[stage] (stage1) at (0,-1.35)
    {Stage One:\\
        Diagonal Screening\\[-1pt]
        (Support Estimation)};

    \node[support] (support) at (0,-2.55)
    {Support estimate $\widehat S$};

    \node[stage] (stage2) at (0,-3.75)
    {Stage Two:\\
        Restricted Dense Estimation\\[-1pt]
        (Eigenvector Extraction)};

    \node[align=center] (output) at (0,-5.15)
    {Initializer $\bfx^0$};

    \draw[mainarrow] (input.south) -- (stage1.north);
    \draw[mainarrow] (stage1.south) -- (support.north);
    \draw[mainarrow] (support.south) -- (stage2.north);
    \draw[mainarrow] (stage2.south) -- (output.north);

    \coordinate (stage1brace) at ($(stage1.east)+(0.03,0)$);
    \coordinate (stage1top) at ($(stage1brace)+(0.28,0.6)$);
    \coordinate (stage1bot) at ($(stage1brace)+(0.28,-0.6)$);
    \draw[brace] (stage1bot) -- (stage1top);
    \node[sidenote, anchor=west] at ($(stage1top)+(0.08,-0.2)$)
    {Order: known screening scaling \cite{Cai2016Optimalrates,Wang2018Sparsephaseretrieval}};
    \node[sidenote, anchor=west] at ($(stage1bot)+(0.08,0.2)$)
    {\textbf{\underline{Constant: less understood}}, only heuristics~\cite{Xu2026Centralizedspectralinitialization}};

    \coordinate (stage2brace) at ($(stage2.east)+(0.03,0)$);
    \coordinate (stage2top) at ($(stage2brace)+(0.28,0.6)$);
    \coordinate (stage2bot) at ($(stage2brace)+(0.28,-0.6)$);
    \draw[brace] (stage2bot) -- (stage2top);
    \node[sidenote, anchor=west] at ($(stage2top)+(0.08,-0.2)$)
    {Order: well explored \cite{Candes2015PhaseretrievalWirtinger, Cai2016Optimalrates}};
    \node[sidenote, anchor=west] at ($(stage2bot)+(0.08,0.2)$)
    {Constant: well understood \cite{Luo2019Optimalspectralinitialization,Mondelli2019Fundamentallimitsweak}};

    \draw[feedback]
    (stage2.west) -- ++(-0.4,0)
    -- ++(0,1.20)
    node[loopinfo, left, xshift=-0.0cm, yshift=-0.5cm]
    {Optional: iterative update\\
        improves \textbf{Order} \\
        via exploiting signal structure}
    -- (support.west);
\end{tikzpicture}%
	}
	\caption{Overview of design questions in sparse phase retrieval initialization. While Stage Two constants and order scaling are well understood, the constant-level spectral design of the stage-one diagonal screening step remains less explored, motivating the deflection-optimal framework developed in this paper.}
	\label{fig:initialization_pipeline}
\end{figure}
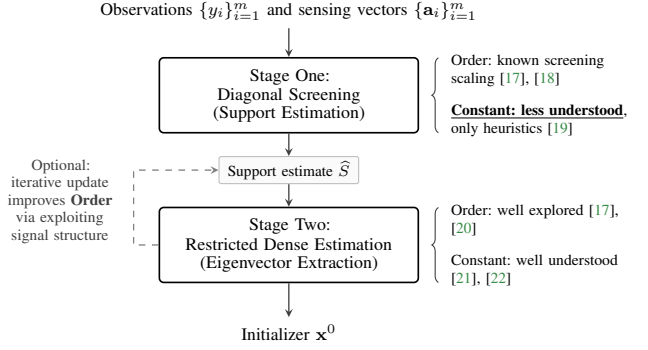

This question has been extensively studied in the sparse phase retrieval literature. Most existing initialization schemes follow a two-stage pipeline, summarized in \Cref{fig:initialization_pipeline}. Stage One forms diagonal scores from a spectral matrix to estimate the support, a procedure closely related to marginal feature screening in high-dimensional statistics~\cite{Fan2008Sureindependencescreening}. Its order-level complexity has been extensively studied~\cite{Cai2016Optimalrates, Wang2018Sparsephaseretrieval}. Stage Two performs restricted dense eigenvector extraction on the selected support, with well-established guarantees at both the order and constant levels~\cite{Candes2015PhaseretrievalWirtinger, Cai2016Optimalrates, Luo2019Optimalspectralinitialization, Mondelli2019Fundamentallimitsweak}. Optional iterative updates may further improve the support estimate at the order-level scaling~\cite{Wu2021HadamardWirtingerflow, Cai2022Sparsesignalrecovery, Xu2024Exponentialspectralpursuit,Cai2023Provablesampleefficienta, Xu2026Achievingoptimalsample}.

This decomposition identifies a critical missing component: a constant-level theory for Stage One. Generally, algorithmic architecture drives order-level improvements by exploiting sparsity, energy concentration, or other signal structure, whereas the choice of spectral matrix governs the coordinate statistics and hence the leading constant. Therefore, characterizing the constant-level performance of stage-one screening is intrinsically tied to its spectral matrix design. This is not only of theoretical interest, but also of practical value, since refining the screening matrix can improve the finite-sample constant with essentially no additional computational cost.

Despite this practical importance, existing spectral matrices are typically chosen for simplicity or proof tractability, including raw~\cite{Wang2018Sparsephaseretrieval}, truncated~\cite{Cai2023Provablesampleefficienta}, or modified constructions~\cite{Gao2017Phaselessrecoveryusing,Xu2024Exponentialspectralpursuit,Xu2026Achievingoptimalsample}. Even the recent centered construction~\cite{Xu2026Centralizedspectralinitialization}, despite its strong empirical performance, relies on heuristic justification. This motivates the question: how should we design the stage-one spectral matrix to optimize the finite-sample screening performance in a principled and systematic manner?

Ideally, such a framework should be derived directly from the final initialization performance. However, the mapping from stage-one spectral design to the end-to-end initialization error is inextricably coupled with the dense estimation dynamics of Stage Two, making direct optimization of the final initialization error analytically prohibitive. To bypass the intractability, we shift our focus to the fundamental statistical properties of the diagonal scores in Stage One. Since each diagonal score, corresponding to a diagonal entry of the spectral matrix, acts as a scalar statistic to distinguish a support coordinate from a non-support coordinate, the design question simplifies to how the spectral preprocessor controls coordinate-wise separation. This perspective naturally leads to the classical deflection criterion from detection theory~\cite{Kay1998Fundamentalsstatisticalsignal, Picinbono1995deflectionperformance, Poor1994introductionsignaldetection,Picinbono1988Optimallinearquadraticsystems}: it maximizes the squared support/non-support mean gap normalized by the non-support variance. As a result, this criterion transforms stage-one spectral design from a heuristic choice of matrices into a principled functional optimization problem.

Building on this criterion, we develop a spectral design theory for stage-one diagonal screening and examine its impact on finite-sample analysis. The specific contributions are threefold:
\begin{enumerate}[label=(\roman*),itemsep=0.2em, topsep=0.2em]
	\item \textbf{Deflection-Optimal Spectral Design:} We establish a principled framework for stage-one diagonal screening. In the Gaussian setting, we characterize the unique optimizer of the exact deflection criterion and analyze its structural properties (\Cref{thm:optimal_preprocessor}). To address the unbounded nature of the Gaussian optimum, we introduce a tuning-free spherical normalization and derive its exact optimal design (\Cref{thm:spherical_jacobi}). Since the exact spherical optimum exhibits a boundary singularity, we construct a bounded surrogate preprocessor and establish its unique optimality under a surrogate deflection criterion (\Cref{prop:spherical_surrogate_optimum}).
	\item \textbf{Finite-Sample Screening Guarantees:} We connect the deflection-optimal design principle to finite-sample performance. For the bounded spherical construction, we establish that the required sample complexity depends explicitly on the deflection quotients, bridging the statistical screening criterion with practical support recovery behavior (\Cref{thm:finite_sample_local_response}).
	\item \textbf{Practical Implementation:} We develop a fully implementable two-stage initializer. We prove that substituting the unknown signal energy with a data-driven estimator preserves the screening behavior up to lower-order effects (\Cref{thm:plugin_stability}), yielding an algorithm consisting of spherical screening followed by restricted dense spectral extraction (\Cref{alg:practical_two_stage}).
\end{enumerate}

Numerical experiments complement the theory by evaluating the centered design principle and its spherical realization at both the screening and end-to-end initialization levels.

The remainder of the paper is organized as follows. \Cref{sec:preliminaries} formalizes the Gaussian and spherical diagonal statistics and sets up the deflection-based design framework. \Cref{sec:results} presents the main theoretical results, comprising the optimal spectral designs, their finite-sample screening guarantees, and practical implementation. \Cref{sec:proofs} proves the design optimality results, while the remaining technical proofs are provided in the appendices. \Cref{sec:discussion} discusses the scope of the stage-one theory. Finally, \Cref{sec:sim} provides numerical results, and \Cref{sec:conclusion} concludes the paper.

\section{Preliminaries and Stage One Design Setup}\label{sec:preliminaries}
Building on the pipeline in \Cref{fig:initialization_pipeline}, we develop the statistical framework for stage-one screening design. We first characterize the Gaussian and spherical diagonal screening statistics, then formulate the associated deflection-based design criteria, and finally express the resulting optimization problems in matched Hilbert-space coordinates.

Throughout the paper, uppercase \(Y\) and \(T\) denote the scalar sensing statistics associated with the Gaussian and spherical models, respectively, while lowercase \(y\) and \(t\) denote deterministic arguments of the corresponding preprocessors. Unless otherwise stated, we adopt the normalization \(\|\bfx\|_2=1\), which is standard in phase retrieval~\cite{Luo2019Optimalspectralinitialization,Candes2015Phaseretrieval,Wu2021HadamardWirtingerflow}. This normalization is used for notational convenience; the general case follows by rescaling, and the unknown signal energy setting is addressed explicitly in \Cref{sec:unknown_norm}.

\subsection{Gaussian Stage-One Family and Diagonal Scores}
Stage One ranks coordinates using diagonal scores extracted from an empirical spectral matrix. We parameterize this matrix by a scalar preprocessor \(f\) through
\begin{equation}\label{eq:Df}
	\bfD_f := \frac{1}{m}\sum_{i=1}^m f(y_i)\bfa_i\bfa_i^*,
\end{equation}
where \(f:[0,\infty)\to\RR\) assigns weights to individual measurements. Preprocessed spectral matrices of the form \eqref{eq:Df} are standard in phase retrieval initialization~\cite{Candes2015Phaseretrieval, Gao2017Phaselessrecoveryusing} and also underlie sparse diagonal-screening methods~\cite{Wang2018Sparsephaseretrieval, Wu2021HadamardWirtingerflow, Cai2023Provablesampleefficienta}. In dense phase retrieval, such matrices are optimized for principal eigenvector extraction~\cite{Luo2019Optimalspectralinitialization, Mondelli2019Fundamentallimitsweak}. Accordingly, we adopt this parameterized family as our stage-one design space, seeking to optimize the preprocessor \(f\) specifically for diagonal screening.

To facilitate this optimization, we first characterize the statistical properties of \(\bfD_f\), which will enter the deflection quotients defined in the sequel. For Gaussian sensing, its expectation admits a highly tractable structure: rotational invariance yields an isotropic background plus a rank-one signal component. Let \(Y:=|\bfa_i^*\bfx|^2\sim\mathrm{Exp}(1)\) under \(\|\bfx\|_2=1\). The resulting decomposition is given as follows.

\begin{proposition}[Gaussian expected response]\label{prop:population}
	For every measurable \(f:[0,\infty)\to\RR\) with \(\EE[f(Y)^2]<\infty\),
	\begin{equation}\label{eq:population_matrix}
		\EE \bfD_f = \alpha_f \bfI_n + \beta_f \bfx\bfx^*
	\end{equation}
	with
	\begin{align}
		\label{eq:gaussian_alpha}\alpha_f & := \EE[f(Y)],      \\
		\label{eq:gaussian_beta}\beta_f   & := \EE[(Y-1)f(Y)].
	\end{align}
\end{proposition}

The proof is given in \Cref{appendices:A}. This decomposition isolates the two quantities relevant to diagonal screening: the isotropic coefficient \(\alpha_f\) represents the background level of the operator, whereas \(\beta_f\) captures the first-order signal gain. Since the background level \(\alpha_f\) is common to all coordinates, subtracting \(\alpha_f\bfI_n\) removes this shared background, motivating the diagonal score
\begin{equation}\label{eq:gaussian_diagonal_score}
	s_j(f) := \bfe_j^*(\bfD_f-\alpha_f\bfI_n)\bfe_j, \qquad j=1,\dots,n.
\end{equation}
By linearity, \Cref{prop:population} implies
\begin{equation}\label{eq:population_score_identity}
	\EE s_j(f) = \beta_f |x_j|^2.
\end{equation}
Since this expected score vanishes for non-support coordinates and equals \(\beta_f |x_j|^2\) for support coordinates, it directly captures the coordinate separation central to stage-one screening, making the statistic \(s_j(f)\) our primary object of study. Because replacing a preprocessor by its negative leaves the design quotients unchanged and flips all scores, we orient the preprocessor so that \(\beta_f\ge 0\) and rank coordinates by the largest scores \(s_j(f)\). This orientation convention applies throughout the paper.

While the mean scores depend only on \(\alpha_f\) and \(\beta_f\), the reliability of support identification depends on the background fluctuations. Let \(S:=\{j:x_j\ne0\}\) denote the support of \(\bfx\). For a non-support coordinate \(j\notin S\), the diagonal score exhibits centered fluctuations with an exact variance
\begin{equation}\label{eq:gaussian_nonsupport_variance}
	\operatorname{Var}(s_j(f))
	=
	\frac{1}{m}\left(2\EE[f(Y)^2]-\alpha_f^2\right).
\end{equation}
This identity follows from \(\EE|a_j|^4=2\) and the independence of the coordinate measurement from the signal projection when \(x_j=0\). Together, the expected response and the null variance fully characterize the Gaussian screening rule, combining to form the deflection quotient in \Cref{sec:deflection}.

\subsection{Spherical Normalization and the Spherical Model}
The unconstrained Gaussian design problem is analytically tractable, but finite-sample concentration for its empirical diagonal matrix must handle the unbounded radial factor \(\|\bfa_i\|_2^2\). This radial component can dominate concentration bounds even though the coordinate information used for support screening is merely directional. We therefore also study a bounded, parameter-free direction-only counterpart. Compared with clipping or truncation, spherical normalization preserves rotational symmetry and introduces no tuning parameter. Below, we characterize the spherical response and variance, which will enter the spherical deflection quotients defined in~\Cref{sec:deflection}.

Define
\begin{equation}\label{eq:spherical_definition}
	t_i := \frac{y_i}{\|\bfa_i\|_2^2},
	\qquad
	\bfu_i := \frac{\bfa_i}{\|\bfa_i\|_2},
\end{equation}
so that \(0\leq t_i\leq 1\), and let
\begin{equation}\label{eq:spherical_general_operator}
	\bfD_\phi
	:=
	\frac{1}{m}\sum_{i=1}^m \phi(t_i)\bfu_i\bfu_i^*.
\end{equation}
Write
\begin{equation}\label{eq:spherical_directional_statistic}
	T:=|\bfu^*\bfx|^2,
\end{equation}
where \(\bfu\) is uniform on the complex unit sphere.

The scalar \(T\) measures the alignment between the sensing direction and the true signal. For \(\bfu\) uniform on the complex unit sphere, \(T\sim\mathrm{Beta}(1,n-1)\), with density
\begin{equation}
	p_n(t)=(n-1)(1-t)^{n-2},
	\qquad 0\le t\le 1,
\end{equation}
and \(\EE T=1/n\). Under spherical normalization, the same structure as in the Gaussian model is preserved, namely, a background component plus a rank-one signal, with the following expected response.

\begin{proposition}[Spherical expected response]\label{prop:spherical_population}
	Let \(\phi:[0,1]\to\RR\) satisfy \(\EE[\phi(T)^2]<\infty\). Then
	\begin{equation}\label{eq:spherical_population}
		\EE \bfD_\phi = \alpha_\phi \bfI_n + \beta_\phi \bfx\bfx^*
	\end{equation}
	with
	\begin{align}
		\label{eq:spherical_alpha} \alpha_\phi
		 & =
		\frac{1}{n-1}\EE[(1-T)\phi(T)], \\
		\label{eq:spherical_beta}\beta_\phi
		 & =
		\frac{n}{n-1}\EE\left[\left(T-\frac{1}{n}\right)\phi(T)\right].
	\end{align}
\end{proposition}

The proof is given in \Cref{appendices:A}. Compared with the Gaussian response, spherical normalization changes the exact forms of $\alpha_\phi$ and $\beta_\phi$, but their roles in diagonal screening remain the same. We therefore define the spherical score
\begin{equation}\label{eq:spherical_diagonal_score}
	s_j(\phi) := \bfe_j^*(\bfD_\phi-\alpha_\phi\bfI_n)\bfe_j, \qquad j=1,\dots,n,
\end{equation}
which satisfies \(\EE s_j(\phi) = \beta_\phi |x_j|^2\) by linearity. The same orientation convention applies.

After the expected contrast has been identified, the remaining ingredient for the deflection criterion is the null fluctuation scale. The spherical counterpart to \eqref{eq:gaussian_nonsupport_variance} is explicit and follows from standard spherical moment identities:
\begin{equation}\label{eq:spherical_nonsupport_variance}
	\operatorname{Var}(s_j(\phi))
	=
	\frac{1}{m}
	\left[
		\frac{2}{n(n-1)}
		\EE\big[(1-T)^2\phi(T)^2\big]
		-
		\alpha_\phi^2
		\right].
\end{equation}
This identity supplies the null fluctuation scale needed by the deflection coefficient, completing the characterization of the spherical screening rule.

\subsection{Stage-One Design Criterion}\label{sec:deflection}
As stated in \Cref{sec:introduction}, we formulate the stage-one design using the deflection coefficient from detection theory~\cite{Kay1998Fundamentalsstatisticalsignal, Picinbono1995deflectionperformance, Poor1994introductionsignaldetection,Picinbono1988Optimallinearquadraticsystems}, which measures the mean separation relative to null fluctuations. For a generic scalar statistic \(z\) used to decide between a null hypothesis \(H_0\) and an alternative \(H_1\), it is defined as
\begin{equation}
	\mathcal D(z;H_1,H_0)
	:=
	\frac{\left(\EE[z\mid H_1]-\EE[z\mid H_0]\right)^2}
	{\operatorname{Var}(z\mid H_0)}.
\end{equation}
In our setting, each diagonal score \(s_j\) plays the role of such a statistic, with \(H_1:j\in S\) and \(H_0:j\notin S\). We therefore design the preprocessor by maximizing the coordinate-wise deflection coefficient over the admissible family.

For a coordinate \(j\in S\), the deflection coefficient of the scalar score \(s_j(f)\) relative to a non-support coordinate \(\ell\notin S\) is given by
\begin{equation}
	\frac{\left(\EE s_j(f)-\EE s_\ell(f)\right)^2}{\operatorname{Var}(s_\ell(f))}.
\end{equation}
For the Gaussian family, \eqref{eq:population_score_identity} and \eqref{eq:gaussian_nonsupport_variance} imply that this quantity admits the factorization
\begin{equation}
	\frac{\left(\EE s_j(f)-\EE s_\ell(f)\right)^2}{\operatorname{Var}(s_\ell(f))}
	=
	m|x_j|^4\, Q_f,
\end{equation}
where
\begin{equation}\label{eq:gaussian_deflection_quotient}
	Q_f
	:=
	\frac{\beta_f^2}{\nu_f},
	\qquad
	\nu_f:=2\EE[f(Y)^2]-\alpha_f^2.
\end{equation}
Thus, \(Q_f\) captures the design-dependent component of the coordinate-wise deflection coefficient. The Gaussian design theorem in Section~\ref{sec:results} is therefore built directly upon maximizing support/non-support mean separation relative to the non-support variance.

For the spherical family, \Cref{prop:spherical_population} and \eqref{eq:spherical_nonsupport_variance} yield the exact spherical deflection quotient
\begin{equation}\label{eq:spherical_exact_deflection_quotient}
	Q_\phi
	:=
	\frac{\beta_\phi^2}{\nu_\phi},
	\qquad
	\nu_\phi
	:=
	\frac{2}{n(n-1)}
	\EE\big[(1-T)^2\phi(T)^2\big]
	-
	\alpha_\phi^2.
\end{equation}
This exact quotient is the direct deflection objective induced by the spherical null variance, serving as the reference design criterion for the spherical model.

\subsection{Hilbert-Space Formulation and Matched Coordinates}\label{sec:orthogonal_bases}
The stage-one design problem is a functional optimization problem: the numerator of each design quotient is a squared linear functional of the weight, while the denominator is a quadratic form determined by the corresponding null fluctuation or energy. We therefore formulate the problem in the \(L^2\) Hilbert space associated with the relevant scalar sensing statistic, \(Y\) in the Gaussian case and \(T\) in the spherical case.

To make this Hilbert-space structure explicit, we introduce matched orthogonal polynomial coordinates. They identify the low-order component that carries the signal gain and separate it from the constant and higher-order components that contribute to the denominator, thereby clarifying the structure of the optimization problems and facilitating the subsequent analysis.

For the Gaussian model, the scalar statistic is \(Y\sim\mathrm{Exp}(1)\), so the matched Hilbert space is
\begin{equation}
	L^2([0,\infty),e^{-y}\rmd y)
	:=
	\left\{
	f:
	\int_0^\infty f(y)^2 e^{-y}\rmd y<\infty
	\right\},
\end{equation}
equipped with inner product
\begin{equation}
	\langle f,g\rangle_{\mathrm G}
	:=
	\int_0^\infty f(y)g(y)e^{-y}\rmd y.
\end{equation}
The orthogonal polynomial family for this exponential weight is the Laguerre basis \(\{L_\ell\}_{\ell\ge 0}\)~\cite{Szego1975Orthogonalpolynomials, Andrews2015Specialfunctions}, characterized by
\begin{equation}
	\int_0^\infty L_\ell(y)L_m(y)e^{-y}\rmd y = 0,
	\qquad \ell\neq m.
\end{equation}
Its first two modes are
\begin{equation}\label{eq:Laguerre_modes}
	L_0(y)=1,
	\qquad
	L_1(y)=1-y.
\end{equation}
With this standard Laguerre convention, the positively oriented linear weight is \(-L_1(y)=y-1\).
Accordingly, every \(f\in L^2([0,\infty),e^{-y}\rmd y)\) admits a mean-square expansion \(f=\sum_{\ell\ge 0} c_\ell L_\ell\). The constant mode \(L_0\) changes only the isotropic offset \(\alpha_f\), while
\begin{equation}
	\beta_f
	=
	\EE[(Y-1)f(Y)]
	=
	\langle y-1,f\rangle_{\mathrm G}
	=
	-\langle L_1,f\rangle_{\mathrm G}.
\end{equation}
Consequently, \(\beta_f\) is determined by the first nonconstant Laguerre projection, whereas \(\EE[f(Y)^2]=\langle f,f\rangle_{\mathrm G} = \sum_{\ell\ge 0} c_\ell^2\) records the total energy of all modes. This coefficient representation isolates the Gaussian design tradeoff: the numerator depends on the first nonconstant projection, while the denominator also accounts for the remaining energy.

The same construction applies to the spherical model. Here the scalar statistic is \(T\sim\mathrm{Beta}(1,n-1)\), with density
\begin{equation}
	p_n(t)=(n-1)(1-t)^{n-2}, \qquad t\in[0,1].
\end{equation}
The matched Hilbert space is
\begin{equation}
	L^2([0,1],p_n(t)\rmd t)
	:=
	\left\{
	\phi:
	\int_0^1 \phi(t)^2 p_n(t)\rmd t<\infty
	\right\},
\end{equation}
equipped with inner product
\begin{equation}
	\langle \phi,\psi\rangle_{\mathrm{sph}}
	:=
	\int_0^1 \phi(t)\psi(t)p_n(t)\rmd t.
\end{equation}

The orthogonal polynomial family associated with the Beta weight \(p_n(t)\propto (1-t)^{n-2}\) is given by the shifted Jacobi system~\cite{Szego1975Orthogonalpolynomials, Andrews2015Specialfunctions}
\begin{equation}
	\left\{
	P_\ell^{(0,n-2)}(1-2t)
	\right\}_{\ell\ge 0},
\end{equation}
obtained by mapping the standard Jacobi polynomials on \([-1,1]\) to \([0,1]\) via \(x=1-2t\). We normalize the first two modes to have unit norm with respect to \(\langle \cdot,\cdot\rangle_{\mathrm{sph}}\). Since \(p_n\) is a probability density, the constant mode is already normalized, and writing
\begin{equation}\label{eq:sph-sigman}
	\sigma_n^2
	:=
	\EE\left[\left(T-\frac{1}{n}\right)^2\right],
\end{equation}
the first two normalized modes are
\begin{equation}\label{eq:Jacobi_modes}
	J_0(t)=1,
	\qquad
	J_1(t)=\frac{t-\frac{1}{n}}{\sigma_n}.
\end{equation}
Since
\begin{equation}\label{eq:beta_phi_Jacobi}
	\beta_\phi
	=
	\frac{n}{n-1}
	\EE\!\left[\left(T-\frac{1}{n}\right)\phi(T)\right]
	=
	\frac{n \sigma_n}{n-1}
	\langle J_1,\phi\rangle_{\mathrm{sph}},
\end{equation}
\(\beta_\phi\) is again determined by the first nonconstant shifted-Jacobi projection.
These matched coordinates provide the coefficient-level representation used to derive the optimal Gaussian and spherical designs in the next section.

\section{Main Results}\label{sec:results}

We present the main results in three steps: we first solve the Gaussian design problem at the level of expectations, then derive a spherical counterpart, and finally establish the corresponding finite-sample guarantees and practical extensions.

\subsection{Unconstrained Gaussian Design}\label{sec:gaussian_design}
We consider the Gaussian stage-one design problem under the exact deflection criterion \(Q_f\) from Section~\ref{sec:preliminaries}. The following theorem characterizes the optimal design.

\begin{theorem}[Gaussian exact deflection optimum]\label{thm:optimal_preprocessor}
	Among all real-valued \(f\in L^2([0,\infty),e^{-y}\rmd y)\) with \(\nu_f>0\), the exact Gaussian deflection quotient \(Q_f\) in \eqref{eq:gaussian_deflection_quotient} is maximized if and only if
	\begin{equation}\label{eq:optimal_preprocessor}
		f^\star(y)=c(y-1)
	\end{equation}
	for some positive constant \(c>0\).
\end{theorem}

\Cref{thm:optimal_preprocessor} completely solves the Gaussian stage-one design problem under the exact deflection criterion, establishing that the centered linear preprocessor is uniquely optimal. The proof is given in \Cref{sec:proof_thm_optimal_preprocessor}. Next, we introduce the spherical counterpart.

\subsection{Spherical Design via Normalization}\label{sec:spherical_design}
To obtain a bounded direction-only redesign, we normalize away the random sensing vector norms and solve the induced design problem defined in \eqref{eq:spherical_definition}-\eqref{eq:spherical_directional_statistic}. In this spherical model, the exact non-support variance quotient remains explicit.

\begin{theorem}[Spherical exact optimum]\label{thm:spherical_jacobi}
	Let \(n\ge 6\). Within the Hilbert space \(L^2([0,1],p_n(t)\rmd t)\), the exact spherical deflection quotient \(Q_\phi\) in \eqref{eq:spherical_exact_deflection_quotient} (over \(\nu_\phi>0\)) is maximized if and only if
	\begin{equation}\label{eq:spherical_exact_deflection_profile}
		\phi^\star(t)
		=
		c\left[
			\frac{t-\frac{1}{n}}{(1-t)^2}
			+
			\frac{1}{(n-2)^2(1-t)}
			\right]
	\end{equation}
	for some positive constant \(c>0\).
\end{theorem}

While \Cref{thm:spherical_jacobi} completes the exact optimal design in the spherical setting, the optimizer \(\phi^\star\) exhibits a boundary singularity at \(t=1\). The requirement \(n\ge 6\) ensures that this singularity remains square-integrable against the Beta measure \(p_n(t)\rmd t\), keeping it strictly within the \(L^2\) space. In practice, since \(t_i = y_i/\|\bfa_i\|_2^2\), a single sensing vector aligned closely with the signal can yield \(t_i \approx 1\), causing the corresponding coordinate score to diverge, posing challenges for finite-sample screening.

To resolve this issue, we propose a bounded surrogate preprocessor:
\begin{equation}\label{eq:canonical_surrogate_phi}
	\phi^\dagger(t) = nt-1.
\end{equation}
This linear preprocessor is smooth, parameter-free, and globally bounded on \([0,1]\). Crucially, \(\phi^\dagger\) is the unique maximizer of the surrogate deflection quotient, defined as follows.

\begin{proposition}[Spherical surrogate optimum]\label{prop:spherical_surrogate_optimum}
	Within the Hilbert space \(L^2([0,1],p_n(t)\rmd t)\), the surrogate quotient
	\begin{equation}\label{eq:spherical_design_quotient}
		\widetilde Q_\phi
		:=
		\frac{\beta_\phi^2}{\gamma_\phi},
		\qquad
		\gamma_\phi := \EE[\phi(T)^2]
	\end{equation}
	(over \(\gamma_\phi>0\)) is maximized if and only if \(\phi^\dagger(t)=c(nt-1)\) for some positive constant \(c>0\).
\end{proposition}

Proposition~\ref{prop:spherical_surrogate_optimum} shows that \(\phi^\dagger\) is the optimal design under the surrogate deflection quotient \(\widetilde{Q}_\phi\), which replaces the exact null variance in the denominator by the \(L^2\)-energy \(\gamma_\phi\). The two optimizers also differ structurally in their Jacobi expansions. The exact optimizer \(\phi^\star\) incorporates higher-order modes to counteract the \((1-t)^2\) weighting in the exact variance denominator, whereas the surrogate optimizer \(\phi^\dagger(t)=nt-1\) retains only the leading nonconstant mode \(J_1\). As confirmed numerically in \Cref{sec:exp3}, these higher-order corrections have negligible impact on the stage-one screening performance. The proofs of \Cref{thm:spherical_jacobi} and \Cref{prop:spherical_surrogate_optimum} are given in \Cref{sec:proof_thm_spherical_jacobi}.

The surrogate optimum \(\phi^\dagger\) also admits a direct connection to the unconstrained Gaussian optimum. Recall from \Cref{thm:optimal_preprocessor} that the Gaussian-optimal preprocessor is \(f^\star(y)=y-1\). Since the spherical normalization discards the radial information and retains only the directional statistic \(T\), we analyze the projection of this Gaussian optimum onto the spherical model. Specifically, the conditional expectation \(\EE[f^\star(Y)\mid T]=\EE[Y-1\mid T]\) characterizes the projection of the Gaussian-optimal preprocessor onto the subspace of functions of the spherical directional statistic. To evaluate this projection, write \(\bfa=R^{1/2}\bfu\) with \(R=\|\bfa\|_2^2\) and \(\bfu=\bfa/\|\bfa\|_2\). By the rotational invariance of \(\mathcal{CN}(0,\bfI_n)\), the squared norm \(R\) is independent of the direction \(\bfu\), with \(\EE[R]=n\). Since \(Y=|\bfa^*\bfx|^2=R\cdot|\bfu^*\bfx|^2=RT\) and \(R\perp T\),
\begin{equation}\label{eq:spherical_projection}
	\EE[Y-1\mid T]=\EE[R]\cdot T-1=nT-1.
\end{equation}
This is exactly the bounded surrogate optimizer \(\phi^\dagger(t)=nt-1\). Thus, the spherical surrogate optimum is not an independent design choice: it is the Gaussian-optimal preprocessor \(y-1\) projected onto the spherical model, inheriting the same first-order structure.

\subsection{A Finite-Sample Diagonal Bridge}\label{sec:finite_sample_bridge}

While the design results in the preceding subsection establish the optimal preprocessors in the population limit, practical diagonal screening must operate under finite measurements. In this subsection, we analyze the finite-sample behavior of the empirical diagonal scores, establishing a theoretical bridge between the statistical deflection quotients and the operational screening performance.

Recall from \eqref{eq:spherical_diagonal_score} that the empirical diagonal scores are given by
\begin{equation}
	s_j
	=
	\bfe_j^*(\bfD_\phi-\alpha_\phi \bfI_n)\bfe_j,
	\qquad
	j=1,\dots,n.
\end{equation}
Let \(\widehat S_k\) denote the indices of the \(k\) largest scores \(s_j\). To characterize the screening performance of \(\widehat S_k\), the following theorem establishes a uniform concentration bound for the diagonal scores, providing a general sample complexity condition to guarantee that the leaked support energy is controlled.

\begin{theorem}[Finite-sample diagonal screening bound]\label{thm:finite_sample_local_response}
	Let \(\phi:[0,1]\to\RR\) satisfy \(\|\phi\|_\infty<\infty\), and let \(Q_\phi\) and \(\widetilde{Q}_\phi\) be the exact and surrogate deflection quotients respectively. For any given energy tolerance \(\epsilon > 0\) and confidence parameter \(\eta>0\), if the sample size satisfies
	\begin{equation}\label{eq:sample_complexity_general}
		m
		\ge
		C_\eta \frac{k^2 \log n}{\epsilon^2} \big(\underbrace{\frac{1}{\sqrt{Q_\phi}}}_{\text{non-support}} + \underbrace{\frac{1}{\sqrt{\widetilde{Q}_\phi}}}_{\text{support}} \big)^2
		+
		\underbrace{C'_\eta \frac{k \|\phi\|_\infty \log n}{\epsilon \beta_\phi}}_{\text{worst-case tail control}},
	\end{equation}
	for sufficiently large constants \(C_\eta, C'_\eta > 0\) depending only on \(\eta\), then with probability at least \(1-3n^{-\eta}\), the leaked support energy of the screened set \(\widehat{S}_k\) satisfies
	\begin{equation}\label{eq:leaked_energy_general}
		\sum_{j\in S\setminus \widehat S_k}|x_j|^2 \le \epsilon.
	\end{equation}
\end{theorem}

Bounding the leaked support energy to a constant level is crucial for initialization: the local convergence of downstream non-convex algorithms only requires the initialization relative error to be bounded by a sufficiently small constant, which is directly controlled by the leaked mass. In other words, ensuring that the leaked energy is below a constant threshold is a necessary and sufficient bottleneck to place the initial estimate within the basin of attraction. Indeed, recent sparse spectral initializers (e.g., \cite{Cai2023Provablesampleefficienta, Xu2024Exponentialspectralpursuit}) implicitly control this leaked energy to establish the final initialization error bounds, serving as the key stepping stone to connect coordinate-wise screening (Stage One) to dense eigenvector extraction (Stage Two).

However, the universal sample complexity bound \eqref{eq:sample_complexity_general} is inherently conservative because it must hold for any arbitrary preprocessor \(\phi\). In the general proof, bounding the coordinate variance of support coordinates requires a worst-case envelope approximation. Directly applying this universal bound to our proposed spherical surrogate \(\phi^\dagger(t) = nt-1\) would introduce an artificial \(n^2\) factor in the sample complexity, resulting in a loose and suboptimal scaling. This artifact arises because the universal bound fails to exploit the concentration properties of our specific optimal preprocessor. To bypass this suboptimality, we can establish a much tighter, localized bound for \(\phi^\dagger\) by exploiting the fact that the random direction component \(T_i\) concentrates closely around its typical scale \(1/n\) with high probability, as formalized in the following corollary.

\begin{corollary}[Localized support energy capture for the spherical surrogate]\label{cor:spherical_surrogate_localized_rate}
	Let \(\phi^\dagger(t)=nt-1\). Under the setup of \Cref{thm:finite_sample_local_response}, for any given energy tolerance \(\epsilon > 0\), if the sample size satisfies
	\begin{equation}\label{eq:sample_complexity_spherical_surrogate}
		m
		\ge
		\underbrace{C_\eta \frac{k^2 \log n}{\epsilon^2}}_{\text{reduced~}\eqref{eq:sample_complexity_general}}
		+
		\underbrace{C'_\eta \frac{k\,\log n\,\log^2(mn)}{\epsilon}}_{\text{localization}},
	\end{equation}
	for sufficiently large constants \(C_\eta, C'_\eta > 0\) depending only on \(\eta\), then with probability at least \(1-3n^{-\eta}\), the leaked support energy of the screened set \(\widehat{S}_k\) satisfies
	\begin{equation}
		\sum_{j\in S\setminus \widehat S_k}|x_j|^2 \le \epsilon.
	\end{equation}
\end{corollary}

The proof of \Cref{cor:spherical_surrogate_localized_rate} (given in \Cref{sec:proof_cor_spherical_surrogate_localized_rate}) formalizes this intuition by restricting the analysis to a high-probability typical event where the spherical summands are localized around their typical \(1/n\) scale. This refined analysis successfully bypasses the worst-case envelope approximations required in the general theorem. Consequently, the leading term in \eqref{eq:sample_complexity_spherical_surrogate} restores the optimal sample complexity scaling of \(k^2 \log n\), aligning with standard sparse recovery literature, while the second term represents a localized tail-control cost.

\subsection{Unknown Signal Energy and a Plug-In Estimator}\label{sec:unknown_norm}

While the preceding subsections assume a known signal energy \(r^2 = \|\bfx\|_2^2\), practical implementation requires estimating \(r^2\) from the measurements. This subsection establishes the robustness of our screening framework under scale estimation, showing that the resulting plug-in error is negligible and does not alter the coordinate ranking. Specifically, we estimate the signal energy \(r^2\) via the empirical mean:

\begin{equation}
	\widehat r^2 := \frac{1}{m}\sum_{i=1}^m y_i, \qquad
	\Delta := \widehat r^2 - r^2.
\end{equation}
For a general signal energy scale, the bounded surrogate optimizer \(\phi^\dagger\) defined in \eqref{eq:canonical_surrogate_phi} yields the oracle spherical preprocessor \(n y_i/\|\bfa_i\|_2^2 - r^2\). Replacing the oracle energy \(r^2\) with its empirical estimate \(\widehat r^2\) yields the practical plug-in matrix. Define
\begin{align}
	\bfD_{\mathrm{sph}}
	 & :=
	\frac{1}{m}\sum_{i=1}^m
	\left(
	\frac{n y_i}{\|\bfa_i\|_2^2}-r^2
	\right)\bfu_i\bfu_i^*,
	\\
	\widehat{\bfD}_{\mathrm{sph}}
	 & :=
	\frac{1}{m}\sum_{i=1}^m
	\left(
	\frac{n y_i}{\|\bfa_i\|_2^2}-\widehat r^2
	\right)\bfu_i\bfu_i^*.
\end{align}
Let \(\alpha_\phi\) and \(\widehat\alpha_\phi\) be the isotropic centering coefficients defined by \eqref{eq:spherical_alpha} under the preprocessors \(\phi(t) = nt - r^2\) and \(\widehat\phi(t) = nt - \widehat r^2\), respectively. In particular, the expectation in \eqref{eq:spherical_alpha} is evaluated solely over the population directional statistic \(T\), treating the empirical scale estimate \(\widehat r^2\) (and thus \(\Delta\)) as a fixed parameter determined by the realized measurements. Since \(\EE[1-T]=(n-1)/n\), this yields
\begin{equation}\label{eq:plugin_alpha_definition}
	\alpha_\phi = r^2 \alpha_{\phi^\dagger},
	\qquad
	\widehat\alpha_\phi = \alpha_\phi - \frac{\Delta}{n},
\end{equation}
where \(\alpha_{\phi^\dagger} = \frac{1}{n-1}\EE[(1-T)(nT-1)]\) is the coefficient for the normalized preprocessor \(\phi^\dagger(t) = nt-1\). We then define the corresponding diagonal scores by
\begin{equation}
	s_j
	:=
	\bfe_j^*(\bfD_{\mathrm{sph}}-\alpha_\phi\bfI_n)\bfe_j,
	\qquad
	\widehat s_j
	:=
	\bfe_j^*(\widehat{\bfD}_{\mathrm{sph}}-\widehat\alpha_\phi\bfI_n)\bfe_j.
\end{equation}
The following theorem establishes this plug-in stability, showing that the estimation error between the oracle and plug-in scores is uniformly controlled.

\begin{theorem}[Plug-in stability for diagonal screening]\label{thm:plugin_stability}
	If \(m\ge \log n\), then for every \(\eta>0\) there exists a constant \(C_\eta>0\) such that with probability at least \(1-4n^{-\eta}\),
	\begin{equation}\label{eq:plugin_uniform_bound}
		\max_{1\le j\le n}
		\left|
		\widehat{s}_j-s_j
		\right|
		\le
		C_\eta r^2\frac{\log n}{mn}.
	\end{equation}
\end{theorem}
The proof is provided in \Cref{sec:proof_thm_plugin_stability}. \Cref{thm:plugin_stability} shows the plug-in error is a negligible perturbation relative to the sorting scores' statistical fluctuations. Specifically, the uniform fluctuation of the oracle scores scales as \(\sqrt{\log n} / (n\sqrt{m})\) (see \eqref{eq:sj-derivation} in the proof of \Cref{cor:spherical_surrogate_localized_rate}), which is of larger order than the perturbation bound \eqref{eq:plugin_uniform_bound}. Consequently, empirical scale substitution does not substantively alter the sample complexity.

While analyzed for the spherical design, this plug-in stability applies generally: global scale substitutions introduce only lower-order perturbations relative to the oracle rule. Since the oracle design is optimal, no empirical plug-in can substantially improve the leading-order coordinate separation beyond it.

\subsection{From Stage-One Design to a Two-Stage Initializer}\label{sec:two-stage-initializer}

The optimal preprocessors and stability guarantees established in the preceding subsections culminate in a practical two-stage initializer, summarized in \Cref{alg:practical_two_stage}. Stage One utilizes the plug-in spherical screening rule to identify the support estimate \(\widehat{S}\). Stage Two then performs restricted dense spectral extraction on \(\widehat{S}\) using the optimal design from~\cite{Luo2019Optimalspectralinitialization}. In practice, the stage-one matrix \(\bfD^{(1)}\) is constructed without the isotropic centering term \(\widehat{\alpha}_\phi\bfI_n\) because subtracting a scalar multiple of the identity shifts all diagonal scores by the same constant, which does not alter the selected top-\(k\) coordinates.

\begin{algorithm}[t]
	\caption{Two-Stage Spectral Initialization with Plug-In Spherical Screening}
	\label{alg:practical_two_stage}
	\begin{algorithmic}[1]
		\REQUIRE Measurements \(\{(\bfa_i,y_i)\}_{i=1}^m\), sparsity level \(k\), threshold parameter \(\veps > 0\).
		\ENSURE Initializer \(\bfx^0\).
		\STATE Estimate the signal energy by \(\widehat{r}^2=\frac{1}{m}\sum_{i=1}^m y_i\).
		\ALGSTAGE{Stage One: Plug-In Spherical Screening}
		\STATE Form the plug-in stage-one matrix:
		\begin{equation*}
			\bfD^{(1)}
			=
			\frac{1}{m}\sum_{i=1}^m
			\left(\frac{n y_i}{\|\bfa_i\|_2^2}-\widehat{r}^2\right)
			\frac{\bfa_i\bfa_i^*}{\|\bfa_i\|_2^2}
		\end{equation*}
		\STATE Compute the diagonal scores \(s_j=\bfe_j^* \bfD^{(1)} \bfe_j\) and let \(\widehat S\) be the \(k\) largest coordinates.
		\ALGSTAGE{Stage Two: Restricted Dense Estimation}
		\STATE Normalize and threshold the measurements: let \(z_i=y_i/\widehat r^2\) and \(\widetilde{z}_i = \max\{z_i, \veps\}\). Form the restricted dense stage matrix:
		\begin{equation*}
			\bfD^{(2)}
			=
			\frac{1}{m}\sum_{i=1}^m
			\left(1-\frac{1}{\widetilde{z}_i}\right)\bfa_i\bfa_i^*
		\end{equation*}
		\STATE Let \(\bfu\) be the principal eigenvector of the principal submatrix \([\bfD^{(2)}]_{\widehat S,\widehat S}\), following the dense spectral construction of~\cite{Luo2019Optimalspectralinitialization}.
		\STATE Set \(\bfx^0\in\CC^n\) by \((\bfx^0)_{\widehat{S}}=\widehat r\bfu\) and \((\bfx^0)_{\widehat{S}^c}=0\).
	\end{algorithmic}
\end{algorithm}

\section{Proofs of Design Results}\label{sec:proofs}

This section proves the two design optimality results. The Gaussian argument reduces the deflection quotient to a two-mode Laguerre decomposition, while the spherical argument uses the corresponding Hilbert-space quotient and Jacobi representation. The finite-sample and plug-in stability proofs are deferred to \Cref{appendices:B}.

\subsection{Proof of \Cref{thm:optimal_preprocessor}}\label{sec:proof_thm_optimal_preprocessor}
Decompose
\begin{equation}
	f = aL_0 + bL_1 + g,
\end{equation}
where \(L_0\) and \(L_1\) are the Laguerre polynomials defined in \eqref{eq:Laguerre_modes}, \(a,b\in\mathbb{R}\), and \(g\perp L_0,L_1\) in \(L^2([0,\infty),e^{-y}\rmd y)\). Since \(L_1(y)=1-y\), this gives
\begin{equation}
	\alpha_f = a,\qquad
	\beta_f = -b,\qquad
	\EE[f(Y)^2] = a^2 + b^2 + \langle g,g\rangle_{\mathrm G},
\end{equation}
and hence
\begin{equation}
	\nu_f
	=
	2\EE[f(Y)^2]-\alpha_f^2
	=
	a^2+2b^2+2\langle g,g\rangle_{\mathrm G}.
\end{equation}
Therefore,
\begin{equation}
	Q_f
	=
	\frac{b^2}{a^2+2b^2+2\langle g,g\rangle_{\mathrm G}},
\end{equation}
which is maximized precisely when \(b\neq0\), \(a=0\), and \(g=0\). Thus, the maximizers are exactly the nonzero scalar multiples of \(L_1\), equivalently \(f^\star(y)=c(y-1)\) with \(c=-b\neq0\). Under the orientation convention \(\beta_f\ge0\), this gives \(c>0\), which proves \eqref{eq:optimal_preprocessor}.~\hfill\IEEEQEDopen

\subsection{Proofs of \Cref{thm:spherical_jacobi} and \Cref{prop:spherical_surrogate_optimum}}\label{sec:proof_thm_spherical_jacobi}
We first prove \Cref{prop:spherical_surrogate_optimum}, which follows the same orthogonal-mode argument used in the Gaussian case. Using the normalized shifted-Jacobi modes \(J_0,J_1\) from \eqref{eq:Jacobi_modes}, decompose
\begin{equation}
	\phi = aJ_0 + bJ_1 + g,
\end{equation}
where \(a,b\in\mathbb{R}\) and \(g\perp J_0,J_1\) in \(L^2([0,1],p_n(t)\rmd t)\). Using \eqref{eq:beta_phi_Jacobi}, we obtain
\begin{equation}
	\beta_\phi = \frac{n\sigma_n}{n-1}b,
	\qquad
	\gamma_\phi = a^2 + b^2 + \langle g,g\rangle_{\mathrm{sph}}.
\end{equation}
Therefore,
\begin{equation}
	\widetilde Q_\phi
	=
	\left(\frac{n\sigma_n}{n-1}\right)^2
	\frac{b^2}{a^2+b^2+\langle g,g\rangle_{\mathrm{sph}}},
\end{equation}
which is maximized precisely when \(b\neq0\), \(a=0\), and \(g=0\). Thus, the maximizers are exactly the nonzero scalar multiples of \(J_1\), equivalently \(\phi^\dagger(t)=c(nt-1)\) with \(c\neq0\). Under the orientation convention \(\beta_\phi\ge0\), this gives \(c>0\), which proves \Cref{prop:spherical_surrogate_optimum}.

We next prove \Cref{thm:spherical_jacobi}. The exact quotient is more involved because, after the change of variables \(g(t)=(1-t)\phi(t)\), its denominator becomes a rank-one perturbation of the \(L^2\) energy rather than the energy itself. We therefore rewrite the optimization as a Hilbert-space quotient. Since \(T\sim \mathrm{Beta}(1,n-1)\) and \(n\ge 6\), the function
\begin{equation}
	h(t):=\frac{t-\frac{1}{n}}{1-t}
\end{equation}
belongs to \(L^2([0,1],p_n(t)\rmd t)\).
Let
\begin{equation}
	g(t):=(1-t)\phi(t).
\end{equation}
Also,
\begin{align}
	\beta_\phi
	 & =
	\frac{n}{n-1}
	\EE\left[\frac{T-\frac{1}{n}}{1-T}g(T)\right],
	\\
	\nu_\phi
	 & =
	\frac{2}{n(n-1)}\EE[g(T)^2]
	-
	\frac{1}{(n-1)^2}\EE[g(T)]^2.
\end{align}
Multiplying the denominator by \(n(n-1)/2\) shows that maximizing \(Q_\phi\) is equivalent to maximizing
\begin{equation}
	\frac{\left\langle h,g\right\rangle_{\mathrm{sph}}^2}
	{\left\langle g,Kg\right\rangle_{\mathrm{sph}}},
\end{equation}
in the Hilbert space \(L^2([0,1],p_n(t)\rmd t)\), where
\begin{equation}\label{eq:jacobi_K_operator}
	K
	:=
	I-\frac{n}{2(n-1)}P_1,
	\qquad
	P_1 g:=\langle g,1\rangle_{\mathrm{sph}} 1.
\end{equation}
Since \(\langle 1,1\rangle_{\mathrm{sph}}=1\) and \(n\ge 6\), the operator \(K\) is positive definite and
\begin{equation}
	K^{-1}
	=
	I+\frac{n}{n-2}P_1.
\end{equation}
The generalized Rayleigh quotient is therefore uniquely maximized by \(g\propto K^{-1}h\). Since \(P_1 h = \langle h,1\rangle_{\mathrm{sph}} 1\), it remains to compute \(\langle h,1\rangle_{\mathrm{sph}}\). Under \(T\sim \mathrm{Beta}(1,n-1)\),
\begin{align}
	\nonumber \left\langle h,1\right\rangle_{\mathrm{sph}}
	 & =
	\EE\left[\frac{T-\frac{1}{n}}{1-T}\right] \\
	 & =
	\EE\left[\frac{T}{1-T}\right]
	-
	\frac{1}{n}\EE\left[\frac{1}{1-T}\right]
	=
	\frac{1}{n(n-2)}.
\end{align}
Hence,
\begin{equation}
	g^\star(t)
	\propto
	h(t)+\frac{n}{n-2}\left\langle h,1\right\rangle_{\mathrm{sph}}
	=
	\frac{t-\frac{1}{n}}{1-t}+\frac{1}{(n-2)^2}.
\end{equation}
Dividing by \(1-t\) gives
\begin{equation}
	\phi^\star(t)
	\propto
	\frac{t-\frac{1}{n}}{(1-t)^2}
	+
	\frac{1}{(n-2)^2(1-t)},
\end{equation}
so the exact maximizers are precisely the nonzero scalar multiples of the profile in \eqref{eq:spherical_exact_deflection_profile}. Under the orientation convention \(\beta_\phi\ge0\), this becomes \eqref{eq:spherical_exact_deflection_profile} with \(c>0\). For \(n\ge 6\), this profile belongs to \(L^2([0,1],p_n(t)\rmd t)\), so the maximizer obtained in the \(g\)-coordinates corresponds to an admissible spherical preprocessor.~\hfill\IEEEQEDopen

\section{Discussion}\label{sec:discussion}

\subsection{Data Dependence in Multi-Stage Initialization}
In the practical implementation of the two-stage initializer (\Cref{alg:practical_two_stage}), reusing the same measurements across both stages introduces statistical dependence between the support screening and restricted dense estimation steps. Consequently, the standard dense optimality results derived under independent measurement assumptions do not directly apply to the coupled procedure, and the dense weight \(1-z_i^{-1}\) is merely a principled dense-stage design rather than a jointly optimal rule. A full joint characterization would require analyzing the conditional distribution induced by the stage-one support selection, which falls outside the scope of the present work. Since our primary objective is the optimal spectral design for stage-one screening, we leave such multi-stage joint optimization for future research.

\subsection{Theoretical Role and Validity of the Surrogate Quotient}
Since \Cref{cor:spherical_surrogate_localized_rate} bypasses the unfavorable scaling of the surrogate quotient \(\widetilde{Q}_\phi\) in \Cref{thm:finite_sample_local_response} through localization, one may ask whether the dependence on both design quotients in \Cref{thm:finite_sample_local_response} can be removed in the same way. This, however, is impossible for the exact quotient \(Q_\phi\). The quantity \(Q_\phi\) is determined by the non-support fluctuations under the null hypothesis, which contain no signal information and therefore cannot benefit from localization. Consequently, \(Q_\phi\) remains the design-dependent factor governing the non-support variance \eqref{eq:nonsupport-j} and the corresponding threshold \eqref{eq:finite_sample_nonsupport} in the concentration bound, making the maximization of the exact deflection quotient necessary for minimizing the overall sample complexity.

Nonetheless, although maximizing \(Q_\phi\) is theoretically necessary, this work optimizes the surrogate quotient \(\widetilde{Q}_\phi\) instead. The justification is that the two optimization problems are asymptotically equivalent, while the surrogate optimizer \(\phi^\dagger\) remains analytically tractable. Indeed, as \(n\to\infty\), the Beta variable \(T \sim \mathrm{Beta}(1,n-1)\) concentrates near \(0\). Consequently, the centering term \(\alpha_\phi^2\) becomes asymptotically negligible, aligning the exact null variance \(\nu_\phi\) with the scaled \(L^2\)-energy \(\gamma_\phi\) via the relation
\begin{equation}
    \nu_\phi \approx \frac{2}{n^2}\gamma_\phi,
    \qquad
    Q_\phi \approx \frac{n^2}{2}\widetilde Q_\phi,
\end{equation}
so that maximizing \(\widetilde Q_\phi\) is asymptotically equivalent to maximizing \(Q_\phi\). However, the exact optimizer \(\phi^\star\) in \eqref{eq:spherical_exact_deflection_profile} exhibits a boundary singularity as \(t\to1\), which complicates finite-sample concentration analysis. In contrast, the surrogate optimizer \(\phi^\dagger\) remains smooth, parameter-free, and globally bounded, leading to the localized guarantee in \Cref{cor:spherical_surrogate_localized_rate}.

\section{Simulations}\label{sec:sim}

\begin{figure*}[t]
	\centering
	\subfloat[Oracle and mismatched stage-one designs.\label{fig:exp1a}]{%
		\includegraphics[width=0.48\textwidth]{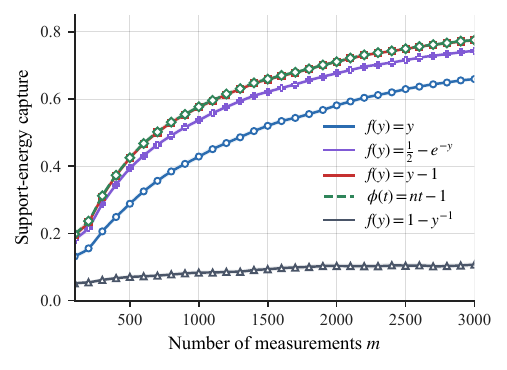}
	}\hfill
	\subfloat[Paired plug-in minus oracle gap.\label{fig:exp1b}]{%
		\includegraphics[width=0.48\textwidth]{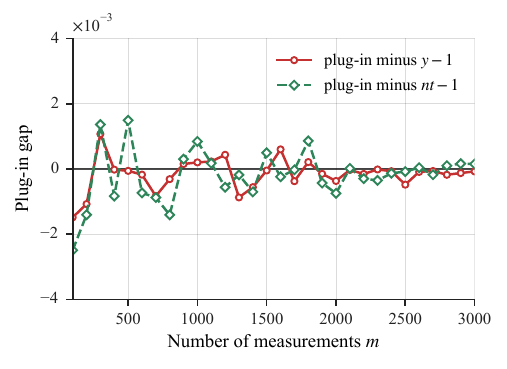}
	}
	\caption{Stage-One support energy capture for \(n=1000\) and \(k=20\). \Cref{fig:exp1a} compares the empirical screening quality of different spectral designs; \Cref{fig:exp1b} shows the paired plug-in loss relative to the oracle version.}
	\label{fig:exp1}
\end{figure*}

This section assesses the practical consequences of the deflection-based design at the level targeted by this paper: spectral initialization. The experiments focus on the initialization module itself, isolating the effect of the stage-one spectral operator without adding a downstream iterative refinement algorithm. \Cref{sec:exp1} compares stage-one support energy capture and plug-in stability across several screening rules. \Cref{sec:exp2} then tests whether the stage-one gains persist in the full two-stage initializer under a common restricted dense extraction step. Finally, \Cref{sec:exp3} compares the bounded spherical surrogate with the exact spherical profile at the finite-sample scale used for screening.

\subsection{Common Setup}
In every trial, a uniformly random support \(S\subset\{1,\dots,n\}\) of size \(k\) is generated, the entries \(\{x_j\}_{j\in S}\) are drawn \iid from \(\mathcal{CN}(0,1)\), \(x_j=0\) is set for \(j\notin S\), and the vector is normalized so that \(\|\bfx\|_2=1\). Measurements are generated as \(y_i=|\bfa_i^*\bfx|^2\) with \(\bfa_i\sim\mathcal{CN}(0,\bfI_n)\). For spherical methods, the normalized statistic \(t_i=y_i/\|\bfa_i\|_2^2\) is used. The dimensions and sample sizes are fixed as \(n=1000,k=20,m\in\{100,200,\dots,3000\},\)
and averages are computed over \(1000\) independent Monte Carlo trials at each sample size \(m\). To ensure a controlled comparison, all methods are evaluated on the same realization of signal and sensing vectors within each trial. The designs compared are summarized in \Cref{tab:stage_one_designs}. In all cases, coordinates are ranked by their respective stage-one diagonal scores, and \(\widehat S_k\) denotes the resulting top-\(k\) support estimate. For \(\bfD_{\mathrm{dense}}\), the preprocessor \(f(y) = 1 - y^{-1}\) is implicitly implemented as \(1 - 1/\max\{y, \veps\}\) with \(\veps = 0.01\) in simulations to prevent the singularity near zero and ensure numerical stability.

\begin{table}[htbp]
	\centering
	\caption{Summary of Stage-One Diagonal Screening Designs}
	\label{tab:stage_one_designs}
	\begin{tabular}{l|c|l}
		\hline
		Operator                                                                            & Model            & Preprocessor Function                                            \\
		\hline
		$\bfD_{y}$~\cite{Wang2018Sparsephaseretrieval, Wu2021HadamardWirtingerflow}         & Gaussian         & $f(y) = y$                                                 \\
		$\bfD_{\exp}$~\cite{Gao2017Phaselessrecoveryusing,Xu2024Exponentialspectralpursuit} & Gaussian         & $f(y) = \frac{1}{2} - e^{-y}$                              \\
		$\bfD_{\mathrm{cen}}$                                                               & Gaussian         & $f(y) = y - 1$                                             \\
		$\bfD_{\mathrm{dense}}$                                                             & Gaussian         & $f(y) = 1 - y^{-1}$                                        \\
		\hline
		$\bfD_{\mathrm{sph}}$                                                               & Spherical        & $\phi(t) = nt - 1$                                         \\
		$\bfD_{\mathrm{sph}}^{\mathrm{exact}}$                                              & Spherical        & $\phi(t) = \frac{t-1/n}{(1-t)^2} + \frac{1}{(n-2)^2(1-t)}$ \\
		\hline
		$\widehat{\bfD}_{\mathrm{cen}}$                                                     & Gaussian (Est.)  & $f(y) = y - \widehat{r}^2$                                 \\
		$\widehat{\bfD}_{\mathrm{sph}}$                                                     & Spherical (Est.) & $\phi(t) = nt - \widehat{r}^2$                             \\
		\hline
	\end{tabular}
\end{table}

\subsection{Stage-One Energy Capture}\label{sec:exp1}

The primary metric for stage-one evaluation is the captured support energy
\begin{equation}
	\mathcal E_{\mathrm{supp}}(\widehat S_k)
	:=
	\frac{\sum_{j\in \widehat S_k\cap S}|x_j|^2}
	{\sum_{j\in S}|x_j|^2}.
\end{equation}
This metric aligns with \Cref{thm:finite_sample_local_response}, as it measures the fraction of signal mass retained within the selected support estimate \(\widehat S_k\).

This experiment assesses whether the design criterion leads to a visible improvement in the screening step. As shown in \Cref{fig:exp1a}, the proposed designs \(\bfD_{\mathrm{cen}}\) and \(\bfD_{\mathrm{sph}}\) yield higher support energy capture than the classical \(\bfD_y\)-based baseline. This separation indicates that the spectral design affects not only the theoretical quotient but also the finite-sample quality of the selected support. Moreover, the performance of the two proposed designs is nearly indistinguishable, which is consistent with the projection identity in \Cref{sec:spherical_design}: the spherical surrogate preserves the structure of the Gaussian design while removing the radial factor.

The deflection criterion also predicts the empirical ordering of standard screening rules. The Gaussian exact quotient yields a closed-form ordering for the canonical Gaussian preprocessors \(f(y)=y\), \(f(y)=\frac{1}{2}-e^{-y}\), and \(f(y)=y-1\). Defining \(Q_f=\beta_f^2/(2\EE[f(Y)^2]-\alpha_f^2)\), it is straightforward to verify that
\begin{equation}
	Q_y=\frac{1}{3} < Q_{\exp}=\frac{3}{8} < Q_{\mathrm{cen}}=\frac{1}{2}.
\end{equation}
This predicts the performance ordering \(\bfD_{\mathrm{cen}} \succ \bfD_{\exp} \succ \bfD_y\), which is consistent with the empirical screening curves in \Cref{fig:exp1a}.

Although the dense spectral preprocessor \(f(y) = 1 - y^{-1}\) is optimal for the stage-two dense eigenvector extraction~\cite{Luo2019Optimalspectralinitialization}, its deflection quotient \(Q_f\) is undefined because the second moment of \(1/Y\) diverges for \(Y \sim \mathrm{Exp}(1)\). In the diagonal screening stage, this infinite variance translates into large coordinate-wise fluctuations induced by small measurements. Consequently, even with the implicit thresholding, \(\bfD_{\mathrm{dense}}\) performs noticeably worse than the deflection-optimal designs \(\bfD_{\mathrm{cen}}\) and \(\bfD_{\mathrm{sph}}\) in \Cref{fig:exp1a}. This contrast highlights that the statistical characteristics of coordinate-wise support screening are fundamentally different from those of dense eigenvector estimation, motivating a dedicated stage-one spectral design.

Finally, \Cref{fig:exp1b} assesses the effect of plug-in estimation. The observed gaps between oracle and plug-in versions are small relative to the performance differences between screening rules, which is consistent with the stability result in \Cref{thm:plugin_stability}.

\subsection{End-to-End Initializer}\label{sec:exp2}
This experiment evaluates the practical significance of the proposed design gain for the complete initializer studied in this paper. After support screening, a restricted dense spectral second stage is applied with the dense spectral preprocessor \(1-z_i^{-1}\), where \(z_i=y_i/\widehat r^2\), following~\cite{Luo2019Optimalspectralinitialization}. Since all methods employ a common second stage, performance variations in \Cref{fig:exp2} isolate the effect of the stage-one support estimate \(\widehat S\). The initialization quality is measured by the phase-invariant relative error of the estimated signal \(\bfx^0\) with respect to the true signal \(\bfx\):
\begin{equation}\label{eq:rel_err_metric}
	\mathcal{L}(\bfx^0, \bfx) := \min_{\theta \in [0, 2\pi)} \frac{\|\bfx^0 - e^{i\theta} \bfx\|_2}{\|\bfx\|_2}.
\end{equation}

\begin{figure}[t]
	\centering
	\includegraphics[width=\columnwidth]{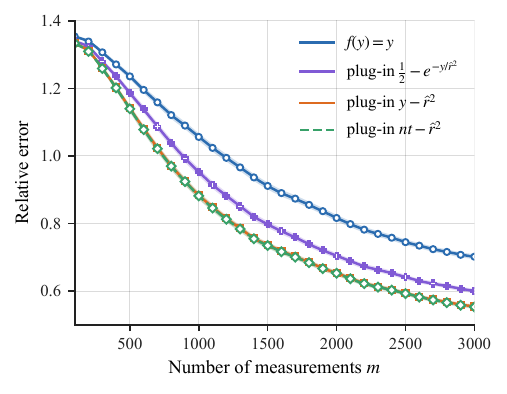}
	\caption{End-to-end relative initialization error under a common dense second stage.}
	\label{fig:exp2}
\end{figure}

The results in \Cref{fig:exp2} show that the screening advantage has a measurable effect on the full two-stage initializer. The centered plug-in variants reduce the relative initialization error compared with the \(\bfD_y\)-based baseline. Specifically, at \(m=3000\), the mean relative error decreases from approximately \(0.70\) for the baseline to \(0.55\) for the two centered variants (Gaussian and spherical), which perform nearly identically. This experiment addresses practical significance at the initialization level: under a fixed dense second stage, better stage-one screening produces a better starting point for downstream refinement methods.

\subsection{Exact Spherical Profile Versus Bounded Surrogate}\label{sec:exp3}
Finally, the approximation underlying \Cref{thm:spherical_jacobi} is tested by comparing the bounded surrogate \eqref{eq:canonical_surrogate_phi} against a truncated implementation of the exact spherical profile \eqref{eq:spherical_exact_deflection_profile}, where \(1-t\) is replaced with \(\max\{1-t,\veps\}\) with \(\veps=0.01\) to avoid the boundary singularity. Six representative sample sizes spanning the full grid are reported.

\begin{table}[h]
	\centering
	\caption{Mean support energy capture comparison (\(\Delta=\mathcal E_{\mathrm{sph}}-\mathcal E_{\mathrm{exact}}\)).}
	\label{tab:exp3}
	\begin{tabular}{c|ccc}
    \hline
    $m$  & surrogate \eqref{eq:canonical_surrogate_phi} & exact profile \eqref{eq:spherical_exact_deflection_profile} & $\Delta\,(\times 10^{-4})$ \\
    \hline
    100  & 0.1543                                     & 0.1536                                              & 6.8                        \\
    300  & 0.3147                                     & 0.3142                                              & 4.8                        \\
    500  & 0.4247                                     & 0.4243                                              & 4.1                        \\
    1000 & 0.5778                                     & 0.5780                                              & -2.2                       \\
    2000 & 0.7109                                     & 0.7105                                              & 4.1                        \\
    3000 & 0.7773                                     & 0.7772                                              & 0.8                        \\
    \hline
\end{tabular}

\end{table}

\Cref{tab:exp3} indicates that the bounded surrogate \eqref{eq:canonical_surrogate_phi} is numerically indistinguishable from the exact profile at the stage-one scale. The paired mean gap \(\Delta\) remains on the order of \(10^{-4}\) throughout the grid \(m \in \{100, \dots, 3000\}\). These results suggest that the bounded surrogate retains the empirical behavior of the exact spherical design while providing a stable and analytically tractable implementation.

\section{Conclusion}\label{sec:conclusion}

This paper develops a deflection-based optimal design perspective on sparse phase retrieval initialization by isolating the matrix used in Stage One from the subsequent dense extraction. By recognizing that these two stages solve fundamentally different estimation problems, we established that the centered linear rule \(y-1\) is the Gaussian exact optimum and that \(nt-1\) is its bounded spherical, direction-only counterpart. This decoupled framework provides a theoretical explanation for the empirical advantages observed for centered preprocessors and shows that improved screening performance can be obtained without introducing additional computational overhead relative to classical baselines.

Beyond establishing the optimal initial matrices under the deflection criterion, this work clarifies that diagonal screening is fundamentally governed by the separation of diagonal scores. This perspective is complementary to recent multistage iterative algorithms: while iterative methods improve the sample-complexity scaling through support updates, our design theory characterizes both the exact and surrogate deflection quotients that appear as the design-dependent factors in the finite-sample screening condition. A natural direction for future work is to adapt these spectral designs to iterative support update steps, which requires overcoming the technical challenges of data-dependent support estimates and conditional geometric structures.

\appendices
\crefalias{section}{appendix}
\crefalias{subsection}{appendix}

\section{Proofs of Supporting Propositions}\label{appendices:A}

\subsection{Proof of \Cref{prop:population}}
\begin{proof}
	Fix a measurable \(f:[0,\infty)\to\RR\) with \(\EE[f(Y)^2]<\infty\), and write
	\(
	\bfM_f:=\EE[f(|\bfa^*\bfx|^2)\bfa\bfa^*]
	\)
	for a generic \(\bfa\sim\mathcal{CN}(0,\bfI_n)\). Since the law of \(\bfa\) is rotationally invariant, \(\bfM_f\) commutes with every unitary matrix that fixes \(\bfx\). Hence, \(\bfM_f\) has the form
	\begin{equation}
		\bfM_f=\alpha \bfI_n+\beta \bfx\bfx^*
	\end{equation}
	for some scalars \(\alpha,\beta\).

	To identify \(\alpha\), let \(\bfu\perp \bfx\) with \(\|\bfu\|_2=1\). Then \(\bfu^*\bfa\sim\mathcal{CN}(0,1)\) is independent of \(\bfa^*\bfx\), so
	\begin{equation}
		\bfu^*\bfM_f\bfu
		=
		\EE[f(|\bfa^*\bfx|^2)|\bfu^*\bfa|^2]
		=
		\EE[f(Y)]\EE[|\bfu^*\bfa|^2]
		=
		\alpha_f.
	\end{equation}
	Hence, \(\alpha=\alpha_f\).

	Next,
	\begin{equation}
		\bfx^*\bfM_f\bfx
		=
		\EE[f(|\bfa^*\bfx|^2)|\bfa^*\bfx|^2]
		=
		\EE[Yf(Y)].
	\end{equation}
	Because \(\|\bfx\|_2=1\), the representation \(\bfM_f=\alpha_f\bfI_n+\beta \bfx\bfx^*\) gives
	\begin{equation}
		\alpha_f+\beta=\EE[Yf(Y)],
	\end{equation}
	so \(\beta=\EE[(Y-1)f(Y)]=\beta_f\). This proves \eqref{eq:population_matrix}.

	Finally,
	\begin{align}
		\EE s_j(f)
		 & =
		\bfe_j^*(\EE \bfD_f-\alpha_f\bfI_n)\bfe_j \nonumber \\
		 & =
		\bfe_j^*(\beta_f \bfx\bfx^*)\bfe_j \nonumber        \\
		 & =
		\beta_f |x_j|^2,
	\end{align}
	which proves \eqref{eq:population_score_identity}.
\end{proof}

\subsection{Proof of \Cref{prop:spherical_population}}
\begin{proof}
	Write \(\bfa=r\boldsymbol{\omega}\), where \(r=\|\bfa\|_2\) and \(\boldsymbol{\omega}=\bfa/\|\bfa\|_2\). For complex Gaussian \(\bfa\), the radius \(r\) and direction \(\boldsymbol{\omega}\) are independent, and \(\boldsymbol{\omega}\) is uniformly distributed on the complex unit sphere. Since
	\begin{equation}
		t
		=
		\frac{|\bfa^*\bfx|^2}{\|\bfa\|_2^2}
		=
		|\boldsymbol{\omega}^*\bfx|^2,
	\end{equation}
	the scalar variable \(T:=t\) depends only on the direction. Standard sphere identities imply
	\begin{equation}
		T\sim \mathrm{Beta}(1,n-1),
	\end{equation}
	with density \(p_n(t)=(n-1)(1-t)^{n-2}\) on \([0,1]\).

	Let
	\begin{equation}
		\bfM_\phi:=\EE[\phi(T)\boldsymbol{\omega}\boldsymbol{\omega}^*].
	\end{equation}
	By rotational invariance, \(\bfM_\phi\) must have the form
	\begin{equation}
		\bfM_\phi = \alpha\bfI_n+\beta\bfx\bfx^*.
	\end{equation}
	To identify the coefficients, choose a unit vector \(\bfv\perp \bfx\). Conditional on \(T=t\), the remaining energy \(1-t\) is uniformly distributed over the orthogonal \((n-1)\)-dimensional subspace, so
	\begin{equation}
		\EE[|\bfv^*\boldsymbol{\omega}|^2\mid T=t]=\frac{1-t}{n-1}.
	\end{equation}
	Therefore,
	\begin{equation}
		\bfv^* \bfM_\phi \bfv
		=
		\frac{1}{n-1}\EE[(1-T)\phi(T)]
		=
		\alpha_\phi.
	\end{equation}

	Along the signal direction,
	\begin{equation}
		\bfx^*\bfM_\phi\bfx
		=
		\EE[T\phi(T)].
	\end{equation}
	Since \(\bfx^*\bfM_\phi\bfx=\alpha_\phi+\beta_\phi\), we obtain
	\begin{align}
		\beta_\phi
		 & =
		\EE[T\phi(T)]
		-
		\frac{1}{n-1}\EE[(1-T)\phi(T)] \nonumber \\
		 & =
		\frac{n}{n-1}\EE\left[\left(T-\frac{1}{n}\right)\phi(T)\right].
	\end{align}
	This establishes \eqref{eq:spherical_population}, \eqref{eq:spherical_alpha}, \eqref{eq:spherical_beta}. Finally,
	\begin{align}
		\EE s_j(\phi)
		 & =
		\bfe_j^*(\EE \bfD_\phi-\alpha_\phi\bfI_n)\bfe_j \nonumber \\
		 & =
		\bfe_j^*(\beta_\phi \bfx\bfx^*)\bfe_j \nonumber           \\
		 & =
		\beta_\phi |x_j|^2,
	\end{align}
	which completes the proof.
\end{proof}

\section{Proofs of Finite-Sample Results}\label{appendices:B}

\subsection{Proof of \Cref{thm:finite_sample_local_response}}\label{sec:proof_thm_finite_sample_local_response}
For \(\eta > 0\), define the non-support and support fluctuation thresholds as
\begin{align}
	\tau_{m,n}^{\mathrm{non}}(\phi;\eta)
	 & :=
	\beta_\phi \sqrt{\frac{2(1+\eta)\log n}{m Q_\phi}}
	+
	\frac{2(1+\eta)\|\phi\|_\infty\log n}{3m},
	\label{eq:finite_sample_nonsupport} \\
	\tau_{m,n}^{\mathrm{supp}}(\phi;\eta)
	 & :=
	\beta_\phi \sqrt{\frac{2(1+\eta)\log n}{m \widetilde{Q}_\phi}}
	+
	\frac{2(1+\eta)\|\phi\|_\infty\log n}{3m}.
	\label{eq:finite_sample_support}
\end{align}
Fix \(j\in\{1,\dots,n\}\). The score deviation \(s_j-\EE s_j = \frac{1}{m}\sum_{i=1}^m \left( \phi(t_i)|u_{ij}|^2 - \EE[\phi(t_i)|u_{ij}|^2] \right)\) is an average of independent, centered random variables, each bounded in absolute value by \(2\|\phi\|_\infty\). Applying Bernstein's inequality~\cite{Vershynin2018Highdimensionalprobabilityintroduction} with a single-sample variance bound \(\operatorname{Var}\left(\phi(t_i)|u_{ij}|^2\right) \le \sigma_j^2\) yields, for any \(u>0\),
\begin{equation}\label{eq:bernstein_general_proof}
	\PP\left( |s_j-\EE s_j| \ge \sqrt{\frac{2 \sigma_j^2 u}{m}} + \frac{2\|\phi\|_\infty u}{3m} \right) \le 2e^{-u}.
\end{equation}
We analyze \(\sigma_j^2\) in two cases depending on the support membership:
\begin{enumerate}[label=(\roman*),itemsep=2pt,topsep=2pt]
	\item \textbf{Non-support coordinates \(j \notin S\):} Here \(x_j=0\), and the exact null variance gives
	      \begin{equation}\label{eq:nonsupport-j}
		      \sigma_j^2 = \beta_\phi^2/Q_\phi.
	      \end{equation}
	      Substituting this into \eqref{eq:bernstein_general_proof} with \(u=(1+\eta)\log n\) and union bounding over all \(n-k\) non-support coordinates yields \(\max_{j \notin S} |s_j| \le \tau_{m,n}^{\mathrm{non}}(\phi;\eta)\) with probability at least \(1-2(n-k)n^{-(1+\eta)}\).

	\item \textbf{Support coordinates \(j \in S\):} Using \(|u_{ij}|^4 \le |u_{ij}|^2 \le 1\), the second moment is bounded by
	      \begin{equation}\label{eq:support-j}
		      \sigma_j^2 \le \EE[\phi(t_i)^2] = \beta_\phi^2/\widetilde{Q}_\phi.
	      \end{equation}
	      Substituting this into \eqref{eq:bernstein_general_proof} with \(u=(1+\eta)\log n\) and union bounding over all \(k\) support coordinates yields \(\max_{j \in S} |s_j - \beta_\phi |x_j|^2| \le \tau_{m,n}^{\mathrm{supp}}(\phi;\eta)\) with probability at least \(1-2kn^{-(1+\eta)}\).
\end{enumerate}
Union bounding both cases, with probability at least \(1-3n^{-\eta}\), statements \eqref{eq:finite_sample_nonsupport} and \eqref{eq:finite_sample_support} hold simultaneously. On this high-probability event, we first establish that any support coordinate with sufficient energy is successfully selected. Indeed, if \(j \in S\) satisfies the screening condition:
\begin{equation}\label{eq:strong_coordinate_condition}
	\beta_\phi |x_j|^2
	>
	\tau_{m,n}^{\mathrm{non}}(\phi;\eta) + \tau_{m,n}^{\mathrm{supp}}(\phi;\eta)
\end{equation}
then \(s_j \ge \beta_\phi\,|x_j|^2 - \tau_{m,n}^{\mathrm{supp}}(\phi;\eta) > \tau_{m,n}^{\mathrm{non}}(\phi;\eta)\). Since all non-support scores satisfy \(s_\ell \le \tau_{m,n}^{\mathrm{non}}(\phi;\eta)\), such \(j\) has a strictly larger score than all non-support coordinates, implying \(j \in \widehat S_k\).

By contraposition, this selection guarantee implies that any missed support coordinate \(j \in S \setminus \widehat S_k\) must violate the condition \eqref{eq:strong_coordinate_condition}, meaning:
\begin{equation}\label{eq:missed_coordinate_energy_bound}
	\beta_\phi |x_j|^2
	\le
	\tau_{m,n}^{\mathrm{non}}(\phi;\eta) + \tau_{m,n}^{\mathrm{supp}}(\phi;\eta).
\end{equation}
Summing this inequality over all missed support coordinates (of which there are at most \(k\)) yields the leaked support energy bound
\begin{equation}\label{eq:leaked_energy_proof_step}
	\sum_{j \in S \setminus \widehat S_k} |x_j|^2
	\le
	\frac{k\left(\tau_{m,n}^{\mathrm{non}}(\phi;\eta) + \tau_{m,n}^{\mathrm{supp}}(\phi;\eta)\right)}{\beta_\phi}.
\end{equation}
Substituting \eqref{eq:finite_sample_nonsupport} and \eqref{eq:finite_sample_support} with \eqref{eq:nonsupport-j} and \eqref{eq:support-j} into \eqref{eq:leaked_energy_proof_step} yields
\begin{align}\nonumber
	\sum_{j \in S \setminus \widehat S_k} |x_j|^2
	 & \le
	\sqrt{2(1+\eta)} \frac{k\sqrt{\log n}}{\sqrt{m}} \left( \frac{1}{\sqrt{Q_\phi}} + \frac{1}{\sqrt{\widetilde{Q}_\phi}} \right) \\
	 & \quad +
	\frac{4(1+\eta)}{3} \frac{k \|\phi\|_\infty \log n}{m \beta_\phi}.
\end{align}
Bounding each of the two terms on the right-hand side by \(\epsilon/2\) yields the sufficient conditions
\begin{align}
	m & \ge \frac{8(1+\eta)k^2\log n}{\epsilon^2}(\frac{1}{\sqrt{Q_\phi}} + \frac{1}{\sqrt{\widetilde{Q}_\phi}})^2,\quad \text{and} \\
	m & \ge \frac{8(1+\eta)k\|\phi\|_\infty\log n}{3\epsilon\beta_\phi}.
\end{align}
Setting \(C_\eta = 8(1+\eta)\) and \(C'_\eta = \frac{8(1+\eta)}{3}\) ensures the leaked support energy is bounded by \(\epsilon\), completing the proof.~\hfill\IEEEQEDopen

\subsection{Proof of \Cref{cor:spherical_surrogate_localized_rate}}\label{sec:proof_cor_spherical_surrogate_localized_rate}
For each coordinate \(j=1,\dots,n\), write
\begin{equation}
	s_j-\EE s_j
	=
	\frac1m\sum_{i=1}^m Z_{ij},
\end{equation}
where
\begin{equation}
	Z_{ij}
	:=
	X_{ij}-\EE X_{ij},
	\qquad
	X_{ij}:=(nt_i-1)|u_{ij}|^2.
\end{equation}
Fix \(C_\eta>0\) large enough and define
\begin{equation}
	\mathcal G_{ij}
	:=
	\left\{
	t_i\le C_\eta \frac{\log(mn)}{n},
	~
	|u_{ij}|^2\le C_\eta \frac{\log(mn)}{n}
	\right\}
\end{equation}
for each sample-coordinate pair \((i,j)\). Both \(t_i\) and \(|u_{ij}|^2\) have the \(\mathrm{Beta}(1,n-1)\) distribution. Thus, for \(B\sim\mathrm{Beta}(1,n-1)\) and \(0\le a\le 1\),
\begin{equation}
	\PP(B>a)
	=
	(1-a)^{n-1}
	\le
	e^{-(n-1)a}.
\end{equation}
Applying this with \(a=C_\eta\log(mn)/n\) to \(t_i\) and \(|u_{ij}|^2\), and using a union bound over the two tails give
\begin{equation}
	\PP(\mathcal G_{ij}^c)
	\le
	(mn)^{-2(1+\eta)}
	\qquad
	\text{for all }i,j.\label{eq:PGijc}
\end{equation}
Set
\begin{equation}
	\widetilde X_{ij}:=X_{ij}\mathbf 1_{\mathcal G_{ij}},
	\qquad
	\widetilde Z_{ij}:=\widetilde X_{ij}-\EE\widetilde X_{ij}.
\end{equation}
For fixed \(j\), \(\{\widetilde Z_{ij}\}_{i=1}^m\) are independent and centered. On \(\mathcal G_{ij}\),
\begin{equation}
	|X_{ij}|
	\le
	(nt_i+1)|u_{ij}|^2
	\le
	C_\eta\frac{\log^2(mn)}{n},
\end{equation}
and hence \(|\widetilde Z_{ij}|\le C_\eta\log^2(mn)/n\). Also,
\begin{align}
	\nonumber
	\operatorname{Var}(\widetilde Z_{ij})
	\le
	\EE X_{ij}^2
	 & =
	\EE\!\left[(nt_i-1)^2|u_{ij}|^4\right]
	\\
	 & \le
	\sqrt{\EE[(nt_i-1)^4]\,\EE|u_{ij}|^8}
	\le
	\frac{C}{n^2}. \label{eq:var-Zij}
\end{align}
Here \(\EE[(nt_i-1)^4]\le C\) and \(\EE|u_{ij}|^8\le C/n^4\). Applying Bernstein's inequality and the union bound over \(j\) yields
\begin{equation}
	\max_{1\le j\le n}
	\left|
	\frac{1}{m}\sum_{i=1}^m\widetilde Z_{ij}
	\right|
	\le
	C_\eta\left[
		\sqrt{\frac{\log n}{mn^2}}
		+
		\frac{\log n\,\log^2(mn)}{mn}
		\right].\label{eq:tildeZij-bernstein}
\end{equation}
with probability at least \(1-2n^{-\eta}\). Then we consider the truncation. Let
\begin{equation}
	\mathcal G
	:=
	\bigcap_{i=1}^m\bigcap_{j=1}^n \mathcal G_{ij}.
\end{equation}
The preceding estimate and a union bound imply \(\PP(\mathcal G^c)\le n^{-\eta}\). On \(\mathcal G\), \(X_{ij}=\widetilde X_{ij}\) for all \(i,j\), so
\begin{equation}
	\frac1m\sum_{i=1}^m Z_{ij}
	=
	\frac1m\sum_{i=1}^m \widetilde Z_{ij}
	-
	\EE[X_{ij}\mathbf 1_{\mathcal G_{ij}^c}].
\end{equation}
By Cauchy--Schwarz, we have
\begin{align}
	\nonumber
	\left|\EE[X_{ij}\mathbf 1_{\mathcal G_{ij}^c}]\right|
	 & \overset{\hphantom{\eqref{eq:PGijc},\eqref{eq:var-Zij}}}{\le}
	\sqrt{\EE X_{ij}^2}\sqrt{\PP(\mathcal G_{ij}^c)}
	\\
	 & \overset{\eqref{eq:PGijc},\eqref{eq:var-Zij}}{\le}
	\frac{C_\eta}{n}(mn)^{-(1+\eta)},
\end{align}
which is negligible relative to the Bernstein bound~\eqref{eq:tildeZij-bernstein}. Therefore, applying the union bound over \(j\), we obtain with probability at least \(1-3n^{-\eta}\),
\begin{equation}
	\max_{1\le j\le n}
	|s_j-\EE s_j|
	\le
	C_\eta\left[
		\sqrt{\frac{\log n}{mn^2}}
		+
		\frac{\log n\,\log^2(mn)}{mn}
		\right]. \label{eq:sj-derivation}
\end{equation}

If \(j\in S\setminus \widehat S_k\), then at least one non-support coordinate has score at least \(s_j\). Since \(\EE s_j=\beta_{\phi^\dagger}|x_j|^2\) on \(S\) and \(\EE s_\ell=0\) on \(S^c\),
\begin{align}
	\nonumber
	\beta_{\phi^\dagger}|x_j|^2
	 & \overset{\hphantom{\eqref{eq:sj-derivation}}}{\le}
	|s_j-\beta_{\phi^\dagger}|x_j|^2|
	+
	\max_{\ell\notin S}|s_\ell|
	\\
	 & \overset{\eqref{eq:sj-derivation}}{\le}
	2C_\eta\left[
		\sqrt{\frac{\log n}{mn^2}}
		+
		\frac{\log n\,\log^2(mn)}{mn}
		\right].
\end{align}
Finally,
\begin{equation}
	\beta_{\phi^\dagger}
	=
	\frac{n}{n-1}\EE\!\left[\left(T-\frac1n\right)(nT-1)\right]
	=
	\frac{1}{n+1},
\end{equation}
so summing over the at most \(k\) missed support coordinates gives
\begin{equation}
	\sum_{j\in S\setminus \widehat S_k}|x_j|^2
	\le
	C_\eta\left[
		k\sqrt{\frac{\log n}{m}}
		+
		\frac{k\log n\,\log^2(mn)}{m}
		\right].
\end{equation}
The stated sample-size condition makes the two terms on the right at most \(\epsilon/2\) each, completing the proof.~\hfill\IEEEQEDopen

\subsection{Proof of \Cref{thm:plugin_stability}}\label{sec:proof_thm_plugin_stability}
Since the plug-in preprocessor differs from the oracle preprocessor by the constant \(-\Delta\), one has
\begin{equation}
	\widehat{\bfD}_{\mathrm{sph}}
	=
	\bfD_{\mathrm{sph}}
	-
	\Delta\widehat{\bfSigma},
	\qquad
	\widehat{\bfSigma}
	:=
	\frac{1}{m}\sum_{i=1}^m \bfu_i\bfu_i^*.
\end{equation}
Using the relationship \(\widehat\alpha_{\phi} = \alpha_{\phi} - \Delta/n\) established in \eqref{eq:plugin_alpha_definition}, we obtain
\begin{equation}
	(\widehat{\bfD}_{\mathrm{sph}}-\widehat\alpha_{\phi}\bfI_n)
	-
	(\bfD_{\mathrm{sph}}-\alpha_{\phi}\bfI_n)
	=
	-\Delta\left(\widehat{\bfSigma}-\frac{1}{n}\bfI_n\right).
\end{equation}
Taking the \(j\)-th diagonal entry yields
\begin{equation}
	\widehat{s}_j - s_j = -\Delta\,\bfe_j^*\left(\widehat{\bfSigma}-\frac{1}{n}\bfI_n\right)\bfe_j.
\end{equation}

Since \(y_i/r^2\sim \mathrm{Exp}(1)\), the empirical scale estimate satisfies the standard sub-exponential concentration bound~\cite{Vershynin2018Highdimensionalprobabilityintroduction}
\begin{equation}\label{eq:delta_bound_main}
	\PP\left(
	|\Delta|
	\ge
	C_0 r^2
	\left(
		\sqrt{\frac{u}{m}}+\frac{u}{m}
		\right)
	\right)
	\le
	2e^{-u}
\end{equation}
for an absolute constant \(C_0>0\) and every \(u>0\).

For the diagonal entries of the empirical covariance,
\begin{equation}
	\bfe_j^*\left(\widehat{\bfSigma}-\frac{1}{n}\bfI_n\right)\bfe_j = \frac{1}{m}\sum_{i=1}^m \left( |u_{ij}|^2 - 1/n \right).
\end{equation}
The centered summands have sub-exponential norm bounded by \(K/n\) and variance bounded by \(C_1/n^2\)~\cite{Vershynin2018Highdimensionalprobabilityintroduction}. Bernstein's inequality and a union bound over the \(n\) coordinates yield, for every \(u>0\),
\begin{align}
	\nonumber & \PP\left(
	\max_{1\le j\le n}
	\left|
	\bfe_j^*\left(\widehat{\bfSigma}-\frac{1}{n}\bfI_n\right)\bfe_j
	\right|
	\ge
	C_2
	\left(
		\sqrt{\frac{u}{mn^2}}+\frac{u}{mn}
		\right)
	\right)
	\\
	          & \quad \le 2n e^{-u}.
\end{align}
Taking \(u=(1+\eta)\log n\), intersecting this event with \eqref{eq:delta_bound_main}, and absorbing constants into \(C_\eta\) yields
\begin{align}
	\nonumber & \max_{1\le j\le n}
	\left|
	\widehat{s}_j-s_j
	\right|
	\\
	\le~      &
	C_\eta r^2
	\left(
	\sqrt{\frac{\log n}{m}}+\frac{\log n}{m}
	\right)
	\left(
	\sqrt{\frac{\log n}{mn^2}}+\frac{\log n}{mn}
	\right).
\end{align}
Under the typical regime \(m \ge \log n\), the right-hand side is dominated by \(C_\eta r^2 \frac{\log n}{mn}\), which proves \eqref{eq:plugin_uniform_bound}.~\hfill\IEEEQEDopen

\bibliographystyle{IEEEtran}
\bibliography{bibs/IEEEabrv2025, bibs/main, bibs/mypapers}

\newpage
\onecolumn

\section*{Supplementary Material: Downstream Nonconvex Refinement Simulations}
\addcontentsline{toc}{section}{Supplementary Material: Downstream Nonconvex Refinement Simulations}

This document provides supplementary numerical simulations evaluating the performance of the proposed deflection-optimal spectral initializers when combined with downstream nonconvex refinement solvers. Specifically, we test the classical Thresholded Wirtinger Flow (TWF) algorithm~\cite{Cai2016Optimalrates} in the full space \(\CC^n\).

\subsection*{Algorithm Description}
We seek to recover the \(k\)-sparse signal \(\bfx \in \CC^n\) from magnitude-only measurements \(y_i = |\bfa_i^* \bfx|^2\) by minimizing the global intensity-based loss:
\begin{equation}\label{eq:supp_loss}
	f(\bfz) = \frac{1}{2m} \sum_{i=1}^m \left( |\bfa_i^* \bfz|^2 - y_i \right)^2, \quad \bfz \in \CC^n.
\end{equation}
Starting from the two-stage spectral initializer \(\bfx^{(0)} = \bfx^0\) from \Cref{alg:practical_two_stage}, we update the estimate via gradient descent followed by hard thresholding to enforce \(k\)-sparsity:
\begin{equation}\label{eq:supp_update}
	\bfx^{(t+1)} = \mathcal{H}_k\left(\bfx^{(t)} - \eta \nabla f(\bfx^{(t)})\right),
\end{equation}
where \(\mathcal{H}_k(\bfv)\) is the hard thresholding operator that keeps the \(k\) largest elements of \(\bfv\) in magnitude and sets the rest to zero. The complex Wirtinger gradient is given by:
\begin{equation}\label{eq:supp_grad}
	\nabla f(\bfx^{(t)}) = \frac{1}{m} \sum_{i=1}^m \left( |\bfa_i^* \bfx^{(t)}|^2 - y_i \right) (\bfa_i^* \bfx^{(t)}) \bfa_i.
\end{equation}
We set the step size \(\eta = 0.05\) and run the gradient descent for \(T = 500\) iterations.

\begin{figure}[!h]
	\centering
	\includegraphics[width=0.6\linewidth]{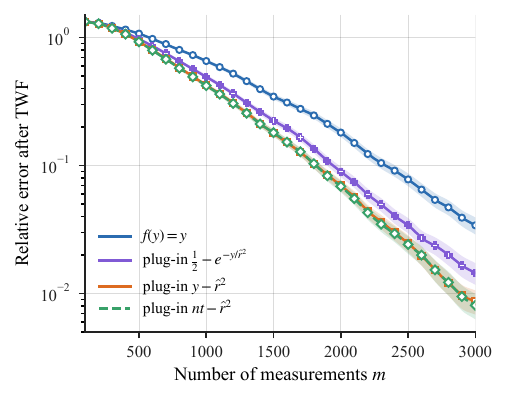}
	\caption{Phase-invariant relative error after downstream Thresholded Wirtinger Flow refinement (500 iterations, step size \(\eta=0.05\)) under different stage-one preprocessors. The simulation uses \(n=1000, k=20\), and the results are averaged over trials. The deflection-optimal designs (centered and spherical surrogate) enable the nonconvex solver to converge to much lower errors at smaller sample sizes.}
	\label{fig:exp_downstream_error}
\end{figure}

\subsection*{Simulation Results}
Under the same experimental settings as in the main text (\(n = 1000, k = 20\)), we compare the phase-invariant relative error \(\mathcal{L}(\bfx^{(T)}, \bfx)\) (defined in~\eqref{eq:rel_err_metric}) after the Thresholded Wirtinger Flow refinement across different stage-one designs. The results, plotted in \Cref{fig:exp_downstream_error}, demonstrate that our proposed deflection-optimal spherical surrogate preprocessor, alongside the centered preprocessor, consistently yields initialization points of superior quality.

This initialization quality enables the downstream nonconvex solver to converge to a significantly lower relative error at smaller sample sizes \(m\) compared to raw and exponential preprocessors, illustrating the practical significance of optimal stage-one design for end-to-end recovery.

\clearpage

\end{document}